%% file: paper.tex
\documentclass[12pt,draftclsnofoot,peerreviewca,onecolumn,a4]{IEEEtran}

\usepackage{amsmath, mathrsfs, mathtools}
\usepackage{amsfonts}
\usepackage{amssymb}
\usepackage{color}
\usepackage{multirow}
\usepackage{stmaryrd}

\allowdisplaybreaks

\usepackage{caption}
\usepackage{subcaption}
\usepackage{float}

\usepackage{amsfonts}
\usepackage{amssymb}
\usepackage{color}
\usepackage{multirow}
\usepackage{graphicx}
\usepackage{tcolorbox}

\newcommand{\subparagraph}{}
\usepackage[compact]{titlesec}

\newcommand{\myhash}{%
	{\settoheight{\dimen0}{C}\kern-.05em\, \resizebox{!}{\dimen0}{\raisebox{\depth}{\#}}}}

\usepackage[numbers,sort&compress]{natbib}

\newcommand{\Sigmah}{{\Sigmam}_{\hv}}
\newcommand{\Sigmay}{{\Sigmay}_{\yv}}

\def\rect{{\ttr\tte\ttc\ttt}}
\def\xim{\boldsymbol{\xi}}

\def\wh{\widehat}
\usepackage{multicol}

\usepackage{algorithm}
\usepackage{algpseudocode}
\usepackage{pifont}
\usepackage{varwidth}


\input{symbols.tex}

\newtheorem{problem}{\bf Problem}

\usepackage{tikz}
\usepackage{pgfplots,tikz-3dplot}
\usetikzlibrary{intersections, pgfplots.fillbetween}
\pgfplotsset{compat=newest}
\usetikzlibrary{patterns}
\usetikzlibrary{positioning}
\usetikzlibrary{datavisualization}
\usetikzlibrary{datavisualization.formats.functions}
\usetikzlibrary{backgrounds}
\usetikzlibrary{shapes,snakes}

\usepgfplotslibrary{external} 
\usepackage{float}
\usepackage{afterpage}
\usepackage{balance}

\def\herm{{\sfH}}

\newcommand{\SigmaS}{\Sigmam_{\hv} (\uv)}
\newcommand{\fvex}{f_{\text{vex}}}
\newcommand{\fcav}{f_{\text{cav}}}

\def\cg{{\clC\clN}} 
\input{macros}

\title{Structured Channel Covariance Estimation from Limited Samples in Massive MIMO}
\author{Mahdi Barzegar Khalilsarai, Tianyu Yang, Saeid Haghighatshoar,  and Giuseppe Caire  
\thanks{The authors are with the Communications and Information Theory Group (CommIT), Technische Universit\"{a}t Berlin (\{m.barzegarkhalilsarai, tianyu.yang, saeid.haghighatshoar, caire\}@tu-berlin.de).}
}

\begin{document}

\maketitle

\def\ful{f_\text{ul}}
\def\fdl{f_\text{dl}}
\def\asfc{\scrC}
\def\asful{\scrC_\text{ul}}
\def\asfdl{\scrC_\text{dl}}

\begin{abstract}
Obtaining channel covariance knowledge is of great importance in various Multiple-Input Multiple-Output MIMO communication applications, including channel estimation and covariance-based user grouping. In a massive MIMO system, covariance estimation proves to be challenging due to the large number of antennas ($M\gg 1$) employed in the base station and hence, a high signal dimension. In this case, the number of pilot transmissions $N$ becomes comparable to the number of antennas and standard estimators, such as the sample covariance, yield a poor estimate of the true covariance and are undesirable. In this paper, we propose a Maximum-Likelihood (ML) massive MIMO covariance estimator, based on a parametric representation of the channel \textit{angular spread function} (ASF). The parametric representation emerges from super-resolving discrete ASF components via the well-known MUltiple SIgnal Classification (MUSIC) method plus approximating its continuous component using suitable limited-support density function. We maximize the likelihood function using a concave-convex procedure, which is initialized via a non-negative least-squares optimization problem. Our simulation results show that the proposed method outperforms the state of the art in various estimation quality metrics and for different sample size to signal dimension ($N/M$) ratios.
\end{abstract}


\section{Introduction}
Knowledge of Uplink (UL) and Downlink (DL) channel covariance matrices of the users yields crucial system-level and computational advantages in MIMO systems and especially in massive MIMO where the number of antennas (thus, the signal dimension) is large ($M\gg 1$) \cite{khalilsarai2018fdd, boroujerdi2018low,haghighatshoar2018low, haghighatshoar2017massive, adhikary2013joint}. Of course, if one has a significantly large number of i.i.d. samples of the user channel vector, one can estimate the covariance matrix precisely. Unfortunately, this is barely the case in massive MIMO since due to the large number of Base Station (BS) antennas the number of i.i.d. channel samples is of the order of the signal dimension. In such scenarios, one requires requires more sophisticated covariance estimators. 

In a more general setup, covariance estimation from limited samples is a classical problem in statistics and is known to be challenging  both statistically and computationally in whenever the signal dimension is large. There is also a vast literature studying the eigenvalue distribution of the sample covariance matrices and in particular their asymptotic dependence on the required sample size and  the signal dimension (see, e.g., \cite{marchenko1967distribution, hachem2005empirical, couillet2011random} and the refs. therein). Apart from classical works, covariance estimation has reemerged recently in many problems in machine learning, compressed sensing, biology, etc. (there is a vast literature; we refer to \cite{pourahmadi2013high, ravikumar2011high, chen2011robust, friedman2008sparse} from some recent results).  What makes these recent works different from the classical ones is the highly-structured nature of the covariance matrices in these applications. For example, it is well-known that the underlying covariance matrices in many applications are sparse or low-rank or satisfy a structure governed by a graphical model. A key challenge in these new applications is to design efficient algorithms, both statistically and numerically, that can take advantage of the underlying structure to recover the covariance matrix with as less sample size and computational resources as possible.

In this paper, we propose a novel covariance estimation algorithm that exploits the structure of MIMO covariance matrices to estimate the underlying covariance matrix with as few number of signal samples as possible. The first step of our method involves a parametric description of the channel \textit{angular spread function} (ASF) in terms of atoms of a carefully designed \textit{dictionary}. This design is based on super-resolving of the \textit{line-of-sight} (LoS) angles of arrival (AoAs) corresponding to discrete ASF components and approximating the continuous ASF component in terms of a family of limited-support density functions (see Section \ref{sec:ASF_modeling}). Super-resolving the discrete ASF components is done via the MUltiple SIgnal Classification (MUSIC) method \cite{stoica2005spectral}. This method is shown to guarantee a consistent estimate of the discrete ASF components, even when the ASF contains continuous components \cite{najim2016statistical}, under mild conditions on the number of antennas $M$, number of samples $N$, and number of discrete components. In the second step of our method, we estimate the parametric ASF by solving a Maximum-Likelihood (ML) problem. This last part is implemented via solving a \textit{Concave-Convex Procedure} (CCCP) with appropriate initialization (see Section \ref{sec:algs}). The advantage of our proposed method lies in both the particular design of the parametric representation of the MIMO covariance and the optimization of the likelihood function.

\section{System Setup}
We consider a BS equipped with a generic array of $M$ antennas and communicating with a set of users. Without loss of generality, we focus on estimating the UL covariance matrix from UL user pilots. We assume orthogonal user pilots in UL and thereby restrict ourselves to the study of UL covariance estimation for a single generic user. We consider the standard block-fading model for the wireless channel (see, e.g., \cite{tse2005fundamentals} and 3GPP channel model)  and denote by $\hv (s)$ the UL channel vector of the user over resource block $s$. We assume that for covariance estimation the BS exploits the user channel vector belonging to a subset of resource blocks where the resource blocks are separated sufficiently in time or in frequency such that the resulting channel vectors are i.i.d. \cite{haghighatshoar2017massive, haghighatshoar2018low}. We denote the set of these i.i.d. samples by
\begin{equation}\label{eq:ch_expression}
\hv (s) = \int_{\Xim} W (\xiv;s) \bfa(\xim) d \xim,~s\in [N],
\end{equation}
where $N$ is the sample size, where $[N]={0,\ldots,N-1}$, where $\Xim=\{ \xiv \in \bR^3 : \Vert \xiv \Vert=1  \}$ is the set of all valid AoAs belonging to the unit sphere, and $\bfa(\xim) \in \bC^M$ denotes the array response vector at $M$ BS antennas. Furthermore, assuming uncorrelated scattering, $W (\xiv;s)$ denotes a zero-mean, stationary, i.i.. Gaussian process over the set of AoAs. The autocorrelation of this process is given as
\begin{equation}\label{eq:w_autocorr}
\bE[W (\xiv;s)W^\ast (\xiv';s)] =  \gamma(\xim) \delta (\xiv - \xiv'),
\end{equation}
where $\gamma (\xiv)$ is the real, positive measure, denoting the ASF. With this definition, the channel covariance matrix can be expressed as
\begin{align}\label{eq:cov_mat}
\Sigmam_\bfh=\bE[\bfh(s)\bfh(s)^\herm]=\int_{\Xim} \gamma(\xim) \bfa(\xim) \bfa(\xim)^\herm d \xim.
\end{align} 
The array response in above formulas is a function of the AoA, antenna location $\rv_i$ (with $i$ denoting the antenna element index) and wavelength $\lambda$ and is given by
\begin{equation}\label{eq:a_vec}
[\av (\xiv) ]_i = e^{j\tfrac{2\pi}{\lambda} \langle \xiv , \rv_i \rangle},
\end{equation}
where $\langle \cdot , \cdot \rangle$ denotes inner product. The pilot signals are received at the BS as $\yv_s = \hv_s x_s + \zv,~s=0,\ldots,N-1$, where $x_s$ is the pilot symbol, assume to take the value $x_s=1\,\forall\, s$ for simplicity and $\zv \sim \cg (\mathbf{0},N_0 \mathbf{I})$ is the additive white Gaussian noise (AWGN). With this setup, we formulate the covariance estimation problem as follows.
\begin{problem}\label{prob:problem_1}
	Given a set of $N$ noisy channel pilot signals $\{ \yv_s\}_{s\in [N]}$, estimate the channel covariance matrix $\Sigmam_\hv$. 
\end{problem}
Throughout this paper we assume that the noise variance $N_0$ is known to the BS. In the large sample regime ($N\to \infty$), problem \ref{prob:problem_1} can be solved using, e.g., the sample covariance 
\begin{equation}\label{eq:sample_cov}
\widehat{\Sigmam}_\hv = \frac{1}{N} \sum_{s=0}^{N-1} \yv_s \yv_s^\herm - N_0 \mathbf{I},
\end{equation}
which yields a consistent estimator. The more interesting regime, however, arises when we consider a limited number of samples, proportional to the channel dimension, i.e. $N \propto M$. In this case, the ``optimal" covariance estimator is generally unknown. In fact, the sample covariance matrix is the maximizer of the covariance likelihood function given samples $\{ \yv_s \}_{s\in[N]}$. But as is well-known, exploiting additional knowledge about the structure of the covariance yields generally better estimates compared to the sample covariance.  Our goal here is to propose such structure in a massive MIMO scenario where the channel covariance is known to belong to the set of MIMO matrices (see \eqref{eq:cov_mat})
\begin{align}\label{mimo_mat2_set}
\clM:=\Big \{\int_{\Xim} \gamma(\xim) \av (\xim) \av (\xim)^\herm d\xim: \gamma \in \Gammam \Big \},
\end{align}
where $\Gammam$ denotes the class of typical ASFs in wireless propagation, e.g., the class of ASFs with relatively small angular support. A structured and yet generic characterization of the ASF is presented next. 

\section{ASF Characterization: Discrete and Continuous Components}\label{sec:ASF_modeling}
Our proposed method hinges upon decomposing the ASF $\gamma(\xiv)$ into its discrete and continuous components as $\gamma(\xiv)=\gamma_d(\xiv)+\gamma_c(\xiv)$ where $\gamma_d(\xiv)$ models the power received from line of sight (LoS) paths and narrow scatterers and where $\gamma_c(\xiv)$ models the power coming from diffuse, wide scatterers. Mathematically, the two components correspond to spikes (or Dirac deltas) and continuous components in the ASF, respectively, written as
\begin{equation}\label{eq:gamma_decomp}
\gamma (\xiv) = \gamma_d (\xiv) + \gamma_c (\xiv) := \sum_{k=1}^{r} c_k \delta (\xiv - \xiv_k) \, +\, \gamma_c (\xiv), 
\end{equation}
where $c_k>0$ for $k=1,\ldots,r$ and $\delta(\cdot)$ denotes Dirac's delta function. Plugging \eqref{eq:gamma_decomp} into \eqref{eq:cov_mat} we obtain a corresponding  decomposition of the channel covariance matrix as
\begin{equation}\label{eq:channel_cov}
\begin{aligned}
\Sigmah =  \Sigmah^d + \Sigmah^c &:=  \sum_{k=1}^{r} c_k \av (\xiv_k) \av (\xiv_k)^\herm\\
& + \int_{\Xim} \gamma_c (\xiv) \av (\xiv) \av (\xiv)^\herm d \xiv,
\end{aligned}
\end{equation}
where $\Sigmah^d$ is a rank-$r$, positive semi-definite (PSD)  matrix corresponding to the covariance matrix of the discrete components and where  $\Sigmah^c$ is a PSD matrix corresponding to the continuous scattering components. Note that since $\gamma_c$ is a continuous distribution, $\Sigmah^c$ is a full-rank matrix (algebraic rank), although in a massive MIMO system it typically has only few significant singular values (low effective rank) when $\gamma_c$ is sparse (namely, it has a limited support in the angular domain). Here is an outline of the steps taken by our proposed algorithm:
\begin{enumerate}

	\item \textit{Spike Location  Estimation for $\gamma_d$:} We use the MUSIC algorithm \cite{stoica2005spectral} to estimate the AoAs of the spike components, i.e., the angles $\{ \xiv_k \}_{k=1}^r$ in \eqref{eq:gamma_decomp}, from  $N$ noisy samples $\{\yv (s)\}_{s\in [N]}$. We will show that for suitable array geometries and under rather mild conditions on the number of spikes $r$, number of antennas $M$, and the number of samples $N$, this method asymptotically yields a consistent estimate of the spike AoAs. However, for complete estimation of the discrete part $\Sigmah^d$, we need to recover the corresponding weights $\{c_k\}_{k=1}^r$, which we do in the next step.
	
	\item \textit{Sparse Dictionary-based Method for Joint Estimation of  $\gamma_d$ and $\gamma_c$:} We assume that the continuous part $\gamma_c$ has a sparse representation over a suitable dictionary consisting of suitable density functions\footnote{By a density function we mean a real positive function $\psi_k (\xiv)$ supported on $\Xim$ with $\int_{-1}^1 \psi_k (\xiv) d \xiv=1$.}  
	\begin{align}\label{muc_dict_2}
	\clG_c:=\{\psi_i(\xiv): i=1,\ldots,n\}.
	\end{align} 
	We then build an over-complete dictionary by combining the dictionary $\clG_c$ with the Delta measures obtained from support estimation of $\gamma_d$ in Step 1 to compose the dictionary $\clG_{c,d}=\clG \cup \{\delta(\xiv-\xiv_k): k=1,\ldots,r\}$. By adopting  dictionary $\clG_{c,d}$, we propose a parametric, finite-dim representation of the (infinite-dim) set of all valid  ASFs $\gamma$ as 
	\begin{align}\label{eq:par_rep}
	\gamma(\xiv)=\sum_{k=1}^r c_k \delta(\xiv - \xiv_k) + \sum_{i=1}^n b_i \psi_i(\xiv),
	\end{align}
	where $\{c_k \}_{k=1}^r$ and $\{b_i\}_{i=1}^n$ are all real and positive. 
	
	It is worthwhile here to state our motivation for support estimation in Step 1. First, without support estimation, we have to use an infinite-dim dictionary $\{\delta(\xiv-\xiv'): \xiv' \in \Xim\}$ to suitably capture the sparsity of the discrete part $\gamma_d$, which makes the estimation problem quite complicated. We could of course neglect these discrete components and the same dictionary $\Gc_c$ for both discrete and continuous parts, but since $\Gc_c$ typically consists of smooth densities, it can not capture the localized nature of the delta functions in $\gamma_d$. This causes a mismatch in covariance estimation and degrades the performance. This is the reason we add support estimation in Step 1 and consider a mixed dictionary for the ASF.

	Given the noisy samples $\{\bfy(s)\}_{s=0}^{N-1}$, we apply the Maximum Likelihood (ML) procedure to  the parametric representation in \eqref{eq:par_rep} to jointly recover the weights $\{c_k\}_{k=1}^r$ and $\{b_i\}_{i=1}^n$ corresponding to the dictionary elements in $\clG_{c,d}$, thus, to estimate the underlying ASFs $\gamma_d$ and $\gamma_c$ by \eqref{eq:par_rep}. We propose a novel method that involves minimizing the non-convex ML cost function with a  suitable initialization based on Non-Negative Least Squares (NNLS). We use the well-known \textit{Concave-Convex Procedure} (CCCP)  \cite{yuille2003concave} to obtain the stationary points of the non-convex ML objective.
	
	\item \textit{From ASF Estimation to Covariance Estimation:} Finally, having estimated $\gamma_d$ and $\gamma_c$ using the ML method, we estimate the covariance $\Sigmah$ via \eqref{eq:channel_cov}.
	\end{enumerate}
	\subsection{Step 1: Discrete ASF Support Estimation}\label{sec:MUSIC}
	In this section, we use the MUSIC method \cite{schmidt1986multiple} to estimate the number $r$ as well as the support $\{\xiv_k: k =1,\ldots,r\}$ of the spikes in the discrete part of the ASF $\gamma_d(\xiv)$. MUSIC was originally proposed for estimating the number and also the frequency of several sinusoids from their mixture contaminated with noise (see, e.g.,  \cite{stoica1989music} and many references therein). We will use the MUSIC in the following form. Let 
	\begin{align}\label{eq:sample_cov_y}
	\widehat{\Sigmam}_\bfy(M)=\frac{1}{N} \sum_{s=0}^{N-1} \bfy(s) \bfy(s)^\herm
	\end{align}
	be the sample covariance of the noisy samples $\{ \yv_s : s\in [N]\}$, where we also denoted the explicit dependence of $\widehat{\Sigmam}_\bfy$ on the signal dimension by $\widehat{\Sigmam}_\bfy(M)$. Let $\widehat{\Sigmam}_\bfy(M)=\widehat{\bfU} \widehat{\Lambdam} \widehat{\bfU}^\herm$ be the eigen-decomposition of $\widehat{\Sigmam}_\bfy(M)$ 
	where $\widehat{\Lambdam}=\diag(\widehat{\lambda}_{1,M}, \dots, \widehat{\lambda}_{M,M})$ denotes the diagonal matrix consisting of the eigenvalues of $\widehat{\Sigmam}_\bfy(M)$, where we assume that the eigenvalues are ordered as $\widehat{\lambda}_{1,M}\geq \dots \geq  \widehat{\lambda}_{M,M}$. The first step of the MUSIC algorithm adapted to our case is to identify the number of spikes $r$. One approach to do this is to find the  index at which there is a significant jump or separation between consecutive eigenvalues. 
	
	Fig. \ref{fig:eig_vec} illustrates this separation between the eigenvalues of the sample covariance matrix $\widehat{\Sigmam}_\bfy(M)$ for different number of antennas $M=25,~50,~100$ with fixed channel dimension to sample size ratio $\frac{M}{N}=\frac{1}{2}$ for an example ASF 
	\begin{equation}\label{eq:ex_asf}
	\gamma (\xi) = \rect_{[-0.7,-0.4]} +\rect_{[0,0.6]} + (\delta (\xi+0.2)+\delta (\xi-0.4))/2,
	\end{equation}
	where $\rect_{\Ac} = \mathbf{1}_{\Ac}$ is the unit-modulus rectangular function over the interval $\Ac$. This ASF contains $r=2$ spikes and two rectangular continuous components. Also the SNR is set to $20$ dBs. For a large enough number of antennas (and even for a moderate number such as $M=25$) the eigenvalue distribution shows a significant jump, such that the two largest eigenvalues ``escape" from the rest. Note that, by increasing the number of antennas, this separation becomes more and more significant, suggesting a way to estimate the number of spikes $r$ in the ASF. Proposition \ref{prop:sep_asymp} of Section \ref{sec:perf_analysis} and the discussion that follows it, rigorously characterize this behavior. 

	\begin{figure*}[t]
		\centering
		\begin{subfigure}[b]{0.31\textwidth}
			\includegraphics[width=\textwidth]{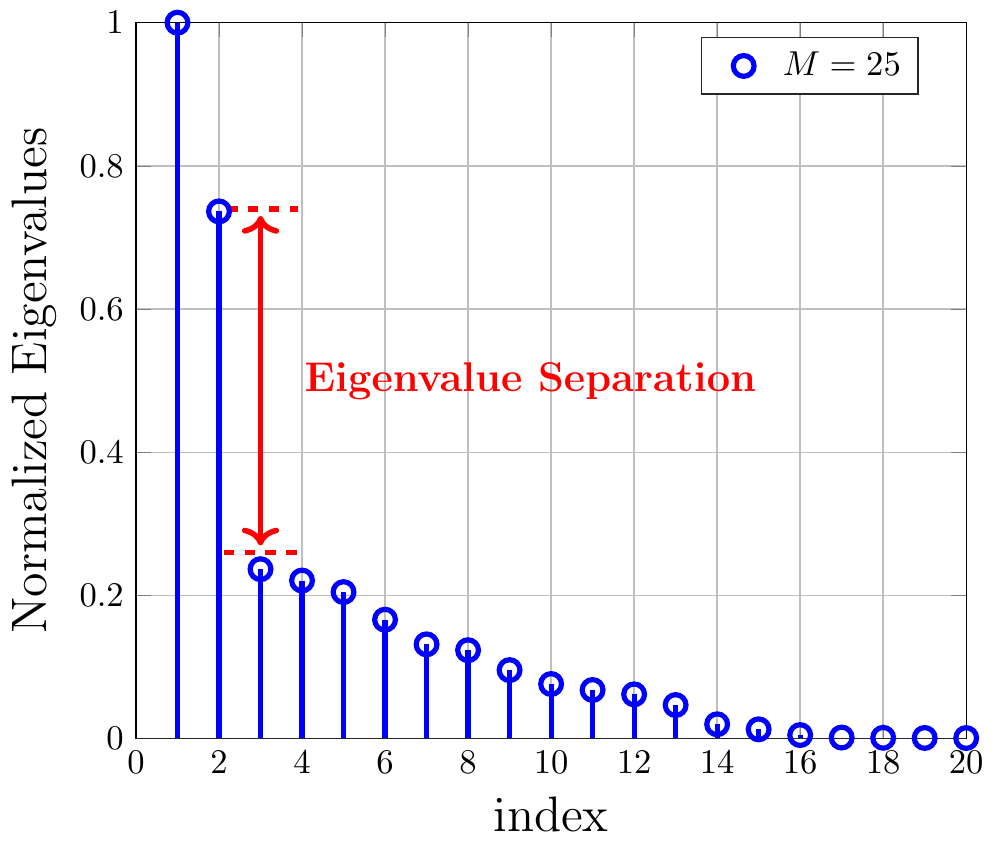}
		\end{subfigure}
		~ 
		\begin{subfigure}[b]{0.31\textwidth}
			\includegraphics[width=\textwidth]{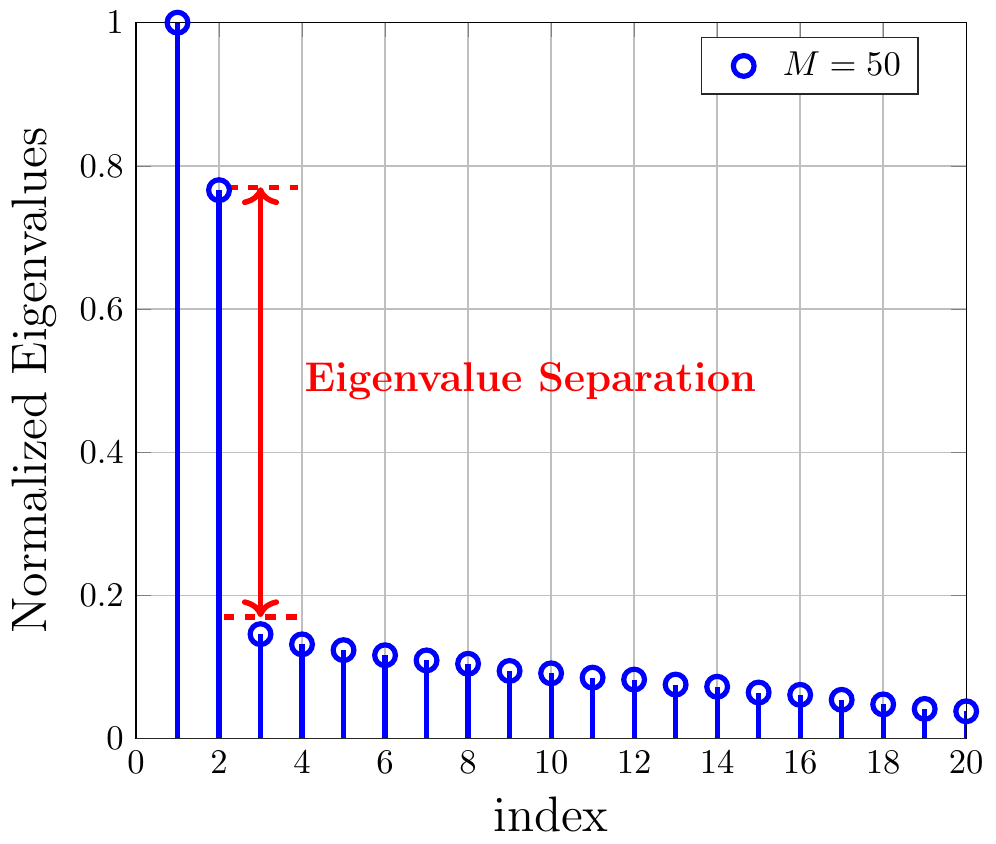}
		\end{subfigure}
		~
		\begin{subfigure}[b]{0.31\textwidth}
		\includegraphics[width=\textwidth]{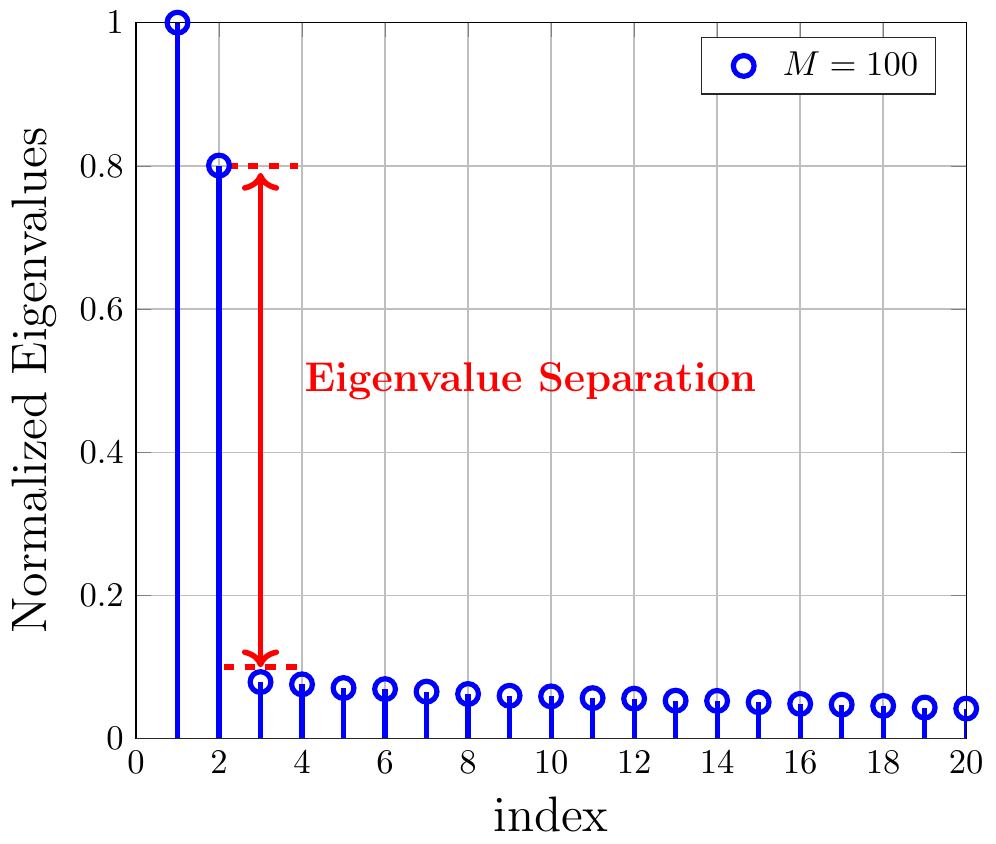}
	\end{subfigure}
\caption{Eigenvalue distribution for the sample covariance matrix $\widehat{\Sigmam}_\bfy(M)$ associated with the example ASF in \eqref{eq:ex_asf} for different values of $M$ and $\frac{M}{N}=\frac{1}{2}$.}\label{fig:eig_vec}
	\end{figure*}

    The separation between the eigenvalues of the sample covariance can be, in principle, characterized in terms of parameters such as the minimum amplitude of the discrete components \cite{najim2016statistical}. Here, since we do not have a priori knowledge about such parameters, we can not calculate the separation size. In other words, we can not calculate a threshold that upperbounds the smaller eigenvalues and separates the from the $r$ largest. Instead, we use a clustering method, and in particular K-means clustering, to separate the group of large eigenvalues from the group of small eigenvalues as follows.
	
	\textbf{Estimating the Number of Spikes $(r)$ via K-means.}
	We first normalize the eigenvalues by the largest eigenvalue $\widehat{\lambda}_{1,M}$ and define the normalized parameters $\beta_i=\big ( \frac{\widehat{\lambda}_{i,M}}{\widehat{\lambda}_{1,M}} \big )^p$ where $p \in (0,1)$ and where $\beta_i\in [0,1]$ for all $i\in \{1,\ldots,M\}$. The role of the exponent $p\in (0,1)$ is to soft-truncate the larger eigenvalues. This can be simply checked from a plot of the function $f:x \mapsto x^p$ in the interval $x \in [0,1]$ where one can see that for $p\in (0,1)$ the function $f$ is quite flat around $x_0=1$ such that, intuitively speaking, all the $x$-values in a large neighborhood of $x_0=1$ are mapped to a very small neighborhood of $f(x_0)=f(1)=1$, thus, soft-truncating of the larger singular values. We use $p=\frac{1}{2}$ for the simulation results illustrated in this paper.
	
	We run K-means clustering algorithm with $K=2$ clusters over the 1-dim set of normalized parameters $\{\beta_i: i =1,\ldots,M\}$ where we initialize the centers of the two clusters with values chosen uniformly at random in the interval $[0,1]$. Denoting by $c^{(\infty)}_{\max}$ and $c^{(\infty)}_{\min}$ the final centers of the clusters after the convergence of the K-means algorithm  and assuming without loss of generality that $c^{(\infty)}_{\min}\leq c^{(\infty)}_{\max}$, we approximate the number of spikes $\widehat{r}$ by the number of  those normalized parameters belonging to the cluster with a larger center $c^{(\infty)}_{\max}$, i.e.,
	\begin{align}
	\widehat{r}=\Big | \big \{i \in \{1,\ldots,M\}: |\beta_i - c^{(\infty)}_{\max}| \leq |\beta_i - c^{(\infty)}_{\min}|  \big \} \Big |. \label{num_spike}
	\end{align}
	We repeat $K$-means clustering several times each time with a different random initialization of the cluster centers. Let us denote the set of all $\widehat{r}$ obtained at different runs by $\{\widehat{r}(\ell): \ell =1,\ldots,L\}$ where $S$ denotes the number of independent runs. We define the empirical CCDF (complementary cumulative density function) of the results by 
	\begin{align}
	F(t)=\frac{\big |  \{\ell: \widehat{r}(\ell) \geq t\} \big |}{L},\  t \in {1,\ldots,M}.
	\end{align}
	Then, we set a threshold $\eta$ very close to $1$ (e.g.,  $\eta=0.95$) and set the final estimate $\widehat{r}$ as 
	\begin{align}
	\widehat{r}= \min\big \{t\in \{1,\ldots,M\}: F(t) \leq 1-\eta\big \}.
	\end{align}
	In this way, we make sure that with an empirical probability at least equal to $\eta$, we have counted all the spikes. Of course, this method may recover \textit{fake} spikes by overestimating the true $r$ (especially when $p\to 0$ such that the larger singular values are too much soft-truncated) but as we will explain later these fake spikes are removed through the proposed algorithm. In other words, it is always better to overestimate the number of spikes than to underestimate them to make sure that one does not miss the true spikes.

	Of course, this method may recover \textit{fake} spikes by overestimating the true $r$ (especially when $p\to 0$ such that the larger eigenvalues are too much soft-truncated) but these fake spikes are removed in the process of the proposed algorithm. In other words, it is always better to overestimate the number of spikes than to underestimate them, to make sure that one does not miss the true spikes. The reason is that, if we estimate more spikes than the true number, just some density elements will be added to our composite dictionary $\Gc_{c,d}$ and it does not harm the eventual covariance estimation.

	Once the number of spikes was estimated through \eqref{num_spike}, MUSIC algorithm proceeds to identify the locations of those spikes. Let $\widehat{\uv}_{\widehat{r}+1,M}, \dots, \widehat{\uv}_{M,M}$ be the eigen-vectors in $\widehat{\bfU}$ corresponding to the smallest $M-\widehat{r}$ eigenvalues and let us define $\bfU_\text{noi}=[\widehat{\uv}_{\widehat{r}+1,M}, \dots, \widehat{\uv}_{M,M}]$ as the $M \times (M-\widehat{r})$ matrix corresponding to the noise subspace. The MUSIC objective function is defined as the 
	\textit{pseudo-spectrum}:
	\begin{equation}\label{eq:pseudo_spectrum_1}
	\widehat{\eta}_M (\xiv) = \|\bfU_\text{noi}^\herm \bfa(\xiv)\|^2=\sum_{k=\widehat{r}+1}^M \left| \av (\xiv)^\herm \widehat{\uv}_{k,M} \right|^2.
	\end{equation}
	MUSIC estimates the support $\{ \widehat{\xiv}_1, \ldots,  \widehat{\xiv}_{\widehat{r}}\}$ of the spikes by identifying $\widehat{r}$ dominant minimizers of $\widehat{\eta}_M (\xiv)$.
		\begin{figure}[t]
		\centering
		\includegraphics[width=0.5\linewidth]{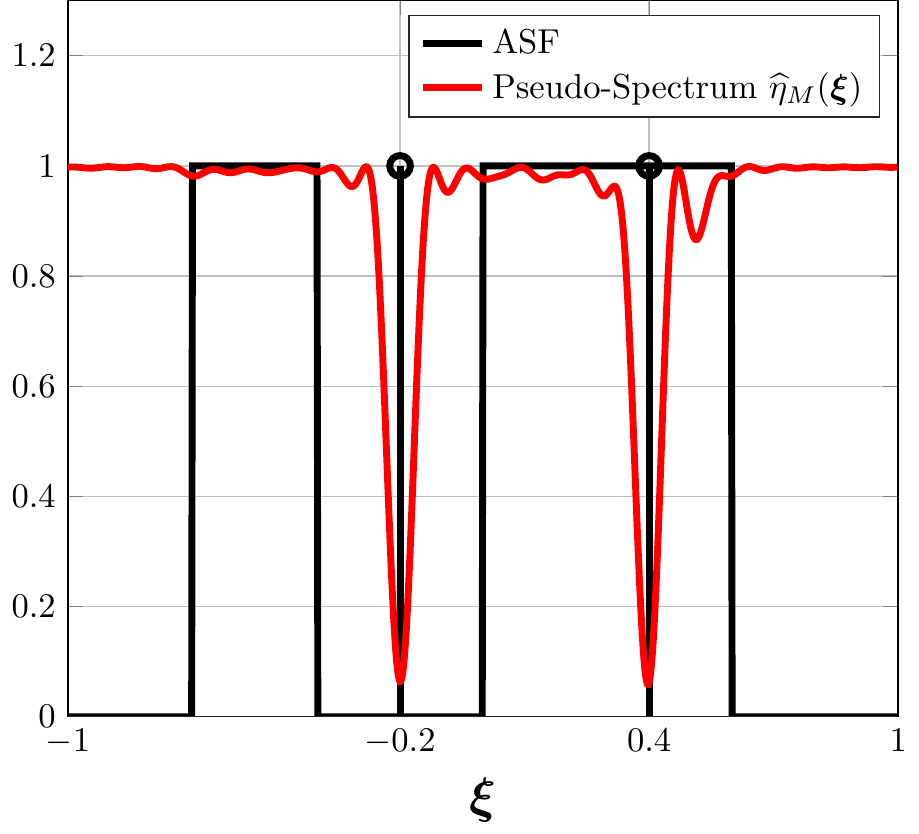}
		\caption{The pseudo-spectrum plotted for the example ASF in \eqref{eq:ex_asf} with $M=25$ and $\frac{M}{N}=1/2$.}
		\label{fig:ASF_MUSIC}
	\end{figure}
	Adapted to our case here, it is well known that when the ASF consists of only a finite number of discrete spikes (no continuous component) and the measurement noise is white Gaussian, the MUSIC estimator implemented as above is consistent \cite{mestre2008improved}. In our case, however, we have also the additional contribution of the continuous (diffuse) part of the ASF.  Mathematically speaking, we can still model this as the recovery of the location of the spikes in noise but the resulting noise will not be white since it will contain the contribution of the diffuse part, which will yield a colored Gaussian vector. Nevertheless, as we will show in Section \ref{sec:perf_analysis}, for large values of $M$ and under some mild condition on the array geometry and also ASF of the continuous part $\gamma_c(\xiv)$, the proposed MUSIC estimator is able to alleviate the colored noise induced by the continuous part. As a result, it still yields a consistent estimate of the location of the spikes in the discrete ASF. Fig \ref{fig:ASF_MUSIC} illustrates the normalized values of the pseudo-spectrum \eqref{eq:pseudo_spectrum_1} for the example ASF in \eqref{eq:ex_asf} and its corresponding sample covariance $\widehat{\Sigmam}_\yv$ for $M=25$ and $\frac{M}{N}=1/2$. As we can see, the $r=2$ smallest minima of the pseudo-spectrum occur very close to the points $\xi=-0.2$ and $\xi=0.4$, which are the locations of the spikes in the true ASF.
	
	A result illustrating the consistency of the MUSIC for such ASFs was proved in \cite{najim2016statistical} for a ULA. This work has been the main inspiration for us to adopt the MUSIC in our setup. We give details on the MUSIC consistency proofs in Section \ref{sec:perf_analysis} to be self-contained. 
	
	Note that we apply MUSIC also to find the spike AoAs in case of a UPA. Unlike the ULA, in this case a consistency proof for the estimator is unclear. However, as we will show in the simulation results section, MUSIC performs quite well in estimating spike locations even in the case of a UPA.

\subsection{Step 2: Sparse Dictionary-Based Method for ML Estimation from Noisy Samples}
Up to this point we have an estimate of the spike locations as $\{ \widehat{\xiv}_k \}_{k=1}^r$ and therefore of the discrete ASF component as 
\begin{align}\label{disc_asf_rep}
\widehat{\gamma}_d (\xiv) = \sum_{k=1}^{\widehat{r}} c_k \delta (\xiv - \widehat{\xiv}_k ),
\end{align}
where the coefficients vector $\cv = [c_1,\ldots,c_{\widehat{r}}]^\transp$ is yet to be estimated. The estimation of the coefficients will be part of the second step in our proposed method, explained below.

In order to estimate the continuous ASF component $\gamma_c (\xiv)$, we consider an approximation of it by the linear combination of the atoms of a specific dictionary of densities $\clG_c = \{ \psi_i (\xiv) : i=1,\ldots,n \}$ as in \eqref{muc_dict_2} as
\begin{equation}
\widehat{\gamma}_c (\xiv) =  \sum_{i=1}^{n} b_i \psi_i (\xiv), \label{cont_asf_rep}
\end{equation}
where $\bv = [b_1,\ldots, b_n]^\transp \in \bR_+^n$
 denotes the non-negative vector of coefficients. This assumption greatly simplifies the estimation task and transforms it from estimating an infinite-dim function $\gamma_c (\xiv)$ to estimating a finite-dim positive vector $\bv$. 
The atoms of the dictionary are selected according to the prior knowledge we have about the propagation environment. Some suitable atoms for the sparse scattering we expect to have wireless channels include localized kernel functions $\psi_i(\xiv)$ , such as  Gaussian, Laplacian, or rectangular kernels,  with a suitably chosen support, or limited-support kernels. Fig. \ref{fig:dic_example} shows an example of rectangular kernel densities for the 1-dim angular domain case (e.g. in the ULA).
	\begin{figure}[t]
	\centering
	\includegraphics[width=0.5\linewidth]{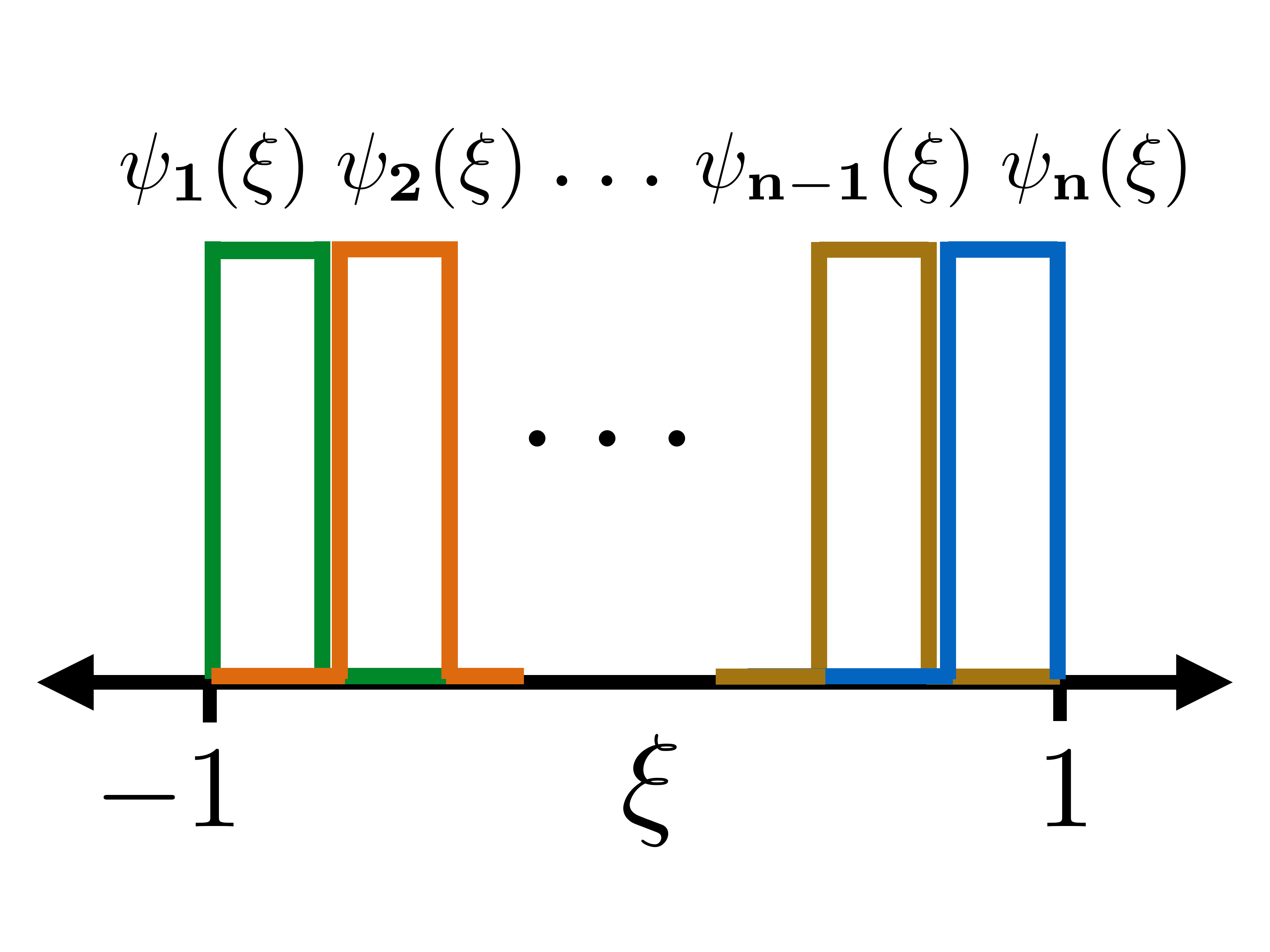}
	\caption{An example of rectangular density functions.}
	\label{fig:dic_example}
\end{figure}
The above decomposition of $\widehat{\gamma}_c (\xiv)$ in \eqref{cont_asf_rep}  and  representation of $\wh{\gamma}_d$ in \eqref{disc_asf_rep} results in the following parametric representation of the channel covariance
\begin{equation}\label{eq:decomp}
\begin{aligned}
\SigmaS &:= \sum_{i=1}^{\wh{r}} c_i \av (\widehat{\xi}_i) \av (\widehat{\xi}_i)^\herm \\
& + \sum_{i=1}^n b_i \int_{\Xim} \psi_i (\xiv) \av (\xiv) \av (\xiv)^\herm d \xiv  
=: \sum_{i=1}^{n+\wh{r}} u_i \Sm_i  \\
\end{aligned}
\end{equation}
where we have defined $\uv = [u_1,\ldots,u_{n+\widehat{r}}]^\transp= $ $[c_1,\ldots,c_{\widehat{r}},b_1,\ldots,b_n]^\transp$, $\Sm_i = \av (\widehat{\xi}_k) \av (\widehat{\xi}_k)^\herm$ for $i=1,\ldots,\wh{r}$ and 
\[\Sm_i =\int_{\Xim} \psi_{i-\wh{r}} (\xiv) \av (\xiv) \av (\xiv)^\herm d \xiv,\] 
for $i=\wh{r}+1,\ldots,\widehat{r}+n$.

\section{Proposed Method: Maximum-Likelihood Estimation}\label{sec:algs}
In order to estimate the coefficients vector $\uv$ from the noisy samples $\{\bfy(s): s\in [N]\}$, we adopt the Maximum Likelihood (ML) estimator. Denoting by $\bfY=[\bfy(0), \dots, \bfy(N-1)]$ the matrix of the observed noisy channel samples, we can write the likelihood function of $\Ym$ given $\SigmaS$ and noise power $N_0$ as
\begin{equation}\label{eq:ML_cals_1}
\begin{aligned}
p \left(\Ym|\SigmaS,N_0 \right) &= \prod_{s=0}^{N-1} p \left(\yv(s)|\SigmaS,N_0 \right)\\
&\hspace{-12mm}= \prod_{s=0}^{N-1} \frac{\exp \left( - \yv(s)^\herm \left( \SigmaS + N_0 \mathbf{I} \right)^{-1} \yv(s) \right)}{\pi^M \det \left( \SigmaS + N_0 \mathbf{I} \right)}\\
&\hspace{-12mm}= \frac{\exp \left( -\tr\left(\left(\SigmaS + N_0 \mathbf{I} \right)^{-1} \Ym \Ym^\herm\right) \right)}{\pi^{MN} \left(\det(\SigmaS + N_0 \mathbf{I})\right)^N }.
\end{aligned} 
\end{equation}
Using \eqref{eq:ML_cals_1} we form the minus log-likelihood function $-\frac{1}{N} \log p \left(\Ym|\uv,N_0 \right)$ and minimize it with respect to the real and non-negative coefficients vector $\uv$. This is formulated as the following optimization problem:


\begin{equation}\label{eq:opt_prog}
\begin{aligned}
&\underset{\uv \in \bR^{n+\wh{r}}_+}{\text{minimize}}~ f(\bfu)&&:=\tr\left( \left(\overset{n+\wh{r}}{ \underset{i=1}{\sum} } u_i \Sm_i  + N_0 \mathbf{I} \right)^{-1}\widehat{\Sigmam}_{\yv}  \right) \\
&~&& + \log \det \left(\overset{n+\wh{r}}{ \underset{i=1}{\sum} }  u_i \Sm_i  + N_0 \mathbf{I} \right),
\end{aligned}
\end{equation}
where $\wh{\Sigmam}_\bfy$ is the sample covariance matrix of the observations in \eqref{eq:sample_cov_y}. Unfortunately, the cost function $f(\bfu)$ in \eqref{eq:opt_prog} is not convex. In fact, it is the sum of a concave and a convex function $f(\uv) = f_{\text{cav}}(\uv)+f_{\text{vex}}(\uv)$ given correspondingly by
\begin{equation}\label{eq:vex_cav}
\begin{aligned}
& f_{\text{cav}}(\uv) =  \log \det \left(\overset{n+\wh{r}}{ \underset{i=1}{\sum} }  u_i \Sm_i  + N_0 \mathbf{I} \right),\\
&f_{\text{vex}}(\uv) =\tr\left( \left(\overset{n+\wh{r}}{ \underset{i=1}{\sum} } u_i \Sm_i  + N_0 \mathbf{I} \right)^{-1}\widehat{\Sigmam}_{\yv}  \right).
\end{aligned}
\end{equation}
Although it is generally difficult to find the global optimum of a non-convex function such as $f(\bfu)$, we use an iterative method that is able to find a \textit{good} stationary point of $f(\bfu)$. Our simulations show that this iterative method is quite efficient and never produces a bad local minimum, but we cannot prove this rigorously. In the following, we explain the mentioned iterative method.

\subsubsection{Optimizing $f(\bfu)$ using Concave-Convex Procedure}
The CCCP \cite{yuille2003concave} is an iterative method for minimizing a  function that is in the form of the sum of a concave and a convex function as in $f(\uv)$, and is guaranteed to converge to a stationary point. The CCCP generates a sequence $\{ \uv^{(t)} : t=1,2,\ldots \}$ as follows: \\
	1) Initialize the problem with a solution $\uv^{(0)}$.\\
	2) At each iteration $t=0,1,\dots$, use the concavity of $\fcav$ to bound it from above by its first order approximation at $\uv^{(t)}$, i.e.,
	\begin{equation}\label{eq:cav_ub}
	\fcav (\uv) \le \fcav (\uv^{(t)})+\langle \nabla\fcav(\uv^{(t)}),\uv - \uv^{(t)}\rangle =: g^{(t)} (\uv),
	\end{equation}
	where $\nabla\fcav$ denotes the gradient of $\fcav$. \\
	3) Use the liner upper bound $g^{(t)} (\uv)$ to $\fcav$ in \eqref{eq:cav_ub} to upper bound $f(\bfu)$ by the proxy function $f^{(t)} (\uv)=\fvex (\uv)+g^{(t)} (\uv)$. Note that $f^{(t)}$ is convex and  tightly approximates $f$  around $\uv = \uv^{(t)}$ tightly. In fact, we have:
	\[ f^{(t)} (\uv^{(t)})=\fvex (\uv^{(t)})+g^{(t)} (\uv^{(t)})= f (\uv^{(t)}).\]
	4) Find the minimizer of the proxy function $f^{(t)}$ and set the next vector in the CCCP iteration to
	\begin{equation}\label{eq:CCP_1}
	\uv^{(t+1)} =\underset{\uv \ge 0}{\arg \min}~ f^{(t)} (\bfu)=\underset{\uv \ge 0}{\arg \min}~ \fvex (\uv)+g^{(t)} (\uv).  
	\end{equation}
Steps (1) to (4) and the Equation \eqref{eq:CCP_1} represent the general iterative CCCP algorithm. Finally, we have the  following theorem, which illustrates a monotonicity property for the CCCP sequence.

\begin{theorem}\label{ccp_thm}
	The sequence $\{ \uv^{(t)}: t=0,1,\dots \}$ produced by CCCP satisfies $f(\uv^{(t+1)})\le f(\uv^{(t)})$. Thus, any limit point of this sequence is a stationary point of the $f(\bfu)$ function. \hfill $\square$
\end{theorem}
\begin{proof}
	Proof follows from Theorem 3 in \cite{yuille2003concave}. 
	\end{proof}
Applying the CCCP to our ML cost function $f(\bfu)$ we reach at the following update rule:
\begin{equation}\label{eq:CCP_2}
\begin{aligned}
\uv^{(t+1)}=& \underset{\uv \geq 0}{\arg\min}~  \tr\left( \left(\overset{n+\wh{r}}{ \underset{i=1}{\sum} } u_i \Sm_i  + N_0 \mathbf{I} \right)^{-1}\widehat{\Sigmam}_{\yv}  \right)+\\
& \tr\left(\overset{n+\wh{r}}{ \underset{j=1}{\sum} } u_j \Sm_j \left(\overset{n+r}{ \underset{i=1}{\sum} } u_i^{(t)} \Sm_i  + N_0 \mathbf{I} \right)^{-1} \right),\\
\end{aligned}
\end{equation}
for a given initial point $\uv^{(0)}$ and for $t=0,1,\ldots$. The optimization in each iteration can be solved by standard convex programming toolboxes, and the algorithm halts when an appropriate convergence property is met. We can also show that the optimization in each iteration is equivalent to a semidefinite program and that it can be solved using, for example, the projected gradient method with desirable convergence properties. These discussions are left for a future work due to space limitations.

\subsection{Initializing the CCCP}
In order to implement the ML method via the CCCP we need a suitable initialization of the optimization variable $\uv$. First, we know that if the number of samples $N$ is large enough, the sample covariance matrix $\widehat{\Sigmam}_\bfy$ would converge to $\Sigmam_\bfy=\Sigmam_\hv + N_0 \bfI$. Moreover, the ML cost function, intuitively speaking, is a metric to find a good fitting to the sample covariance matrix $\widehat{\Sigmam}_\bfy$ from the set of all covariance matrices of the form 
\begin{align}\label{par_cov_last}
\Sigmam_\bfy=\Sigmam_\hv + N_0 \bfI=\sum_{j=1}^{n+\wh{r}} u_j \bfS_j + N_0 \bfI.
\end{align}
Therefore, to initialize the CCCP, we can use a simpler metric to perform this fitting between $\widehat{\Sigmam}_\bfy$ and parametric covariance in \eqref{par_cov_last}. For this purpose, we use the Frobenius norm as a fitting metric and find the initial point for CCCP as 
\begin{align}
\bfu^{(0)}=\argmin_{\bfu \geq 0} \| \widehat{\Sigmam}_\bfy - \sum_{j=1}^{n+\wh{r}} u_j \bfS_j - N_0 \bfI\|_\sfF^2.
\end{align}
Applying vectorization and defining $\bfA=[\vec(\bfS_1) \dots \vec(\bfS_{n+\wh{r}})]$ and $\bff=\vec(\widehat{\Sigmam}_\bfy - N_0 \bfI)$, we can write this as a Non-Negative Least Square (NNLS) problem
\begin{align}\label{nnls_init_ccp}
\bfu^{(0)}=\argmin_{\bfu \geq 0} \|\bfA \bfu - \bff\|^2,
\end{align}
which can be easily solved using standard NNLS solvers. Denoting the solution of the iterative procedure in \eqref{eq:CCP_2} as $\uv^\star$, the channel covariance estimate is given as
\begin{equation}\label{eq:final_estimate}
\Sigmam_\hv^\star = \sum_{i=1}^{n+\wh{r}} u_i^\star \Sm_i,
\end{equation}
which concludes our covariance estimation method.
\section{Simulation Results}
In this section, we perform numerical simulations to compare the performance of our proposed algorithm with the following two rival methods.

 \subsection{The SPICE Method} The first method is known as \textit{sparse iterative covariance-based estimation} (SPICE) \cite{stoica2011spice}. Assuming a dictionary $\Dm =[\av (\xiv_1),\ldots,\av (\xiv_Q)]$ of $Q$ array response vectors corresponding to $Q$ angles of arrival, this method parameterizes the channel covariance matrix as $\Sigmam_\hv = \Dm \text{diag}(\uv) \Dm^\herm$, where the vector $\uv \in \bR_+^{Q}$ represents the channel variance along the set of AoAs $\{ \xiv_1,\ldots,\xiv_Q \}$. Then the parameters vector $\uv$ is estimated by solving the following convex program:
\begin{equation}\label{eq:SPICE}
\begin{aligned}
 \uv^\star = &\underset{\uv}{\arg \min} \Vert \Sigmam^{-1/2} (\widehat{\Sigmam}_\yv - \Sigmam)\widehat{\Sigmam}_\yv^{-1/2} \Vert^2 ~ \\ &\text{subject to} ~  \Sigmam = \Dm \text{diag}(\uv) \Dm^\herm.
\end{aligned}
\end{equation}
The channel covariance estimate is then given as $\Sigmam_\hv^{\text{SPICE}} = \Dm \text{diag}(\uv^\star) \Dm^\herm$. Note that the special parametric modeling of the dictionary $\Dm$ in this method is equivalent to our modeling in \eqref{eq:decomp} only if we assume the family of density functions $\Gc_c$ in \eqref{muc_dict_2} to be consisting of only delta functions, i.e. $\psi_i = \delta (\xiv - \xiv_i)$ for $i=1,\ldots,n$. In this sense our parametric model is more general and allows for a wider choice of density functions to approximate the ASF. In other words, for us the covariance matrix needs not be a linear combination of rank-1 matrices $\av (\xiv_i) \av (\xiv_i)^\herm,~i=1,\ldots, Q$, as is the case for SPICE.  

To have  a better theoretical understanding of SPICE, it is worthwhile here to mention that there is an interesting  relation between the SPICE and our proposed ML method through the Bregman divergences \cite{bregman1967relaxation}, which we briefly explain in the following. Let us consider the strictly convex function $f(\Psim)=-\log | \Psim |$ over the space of $M \times M$ PSD matrices. For two PSD matrix $\Psim_1, \Psim_2$, the Bregman divergence between  $\Psim_1$ and $\Psim_2$ generated by $f(\Psim)$ is defined by 
\begin{align}
\scrD(\Psim_1,& \Psim_2)=f(\Psim_1) - f(\Psim_2) - \inp{\nabla f(\Psim_2)}{\Psim_1-\Psim_2}\nonumber\\
&= \log|\Psim_2| - \log|\Psim_1| + \trace \left [ \Psim_2^{-1} (\Psim_1-\Psim_2) \right ] \nonumber\\
&=\log|\Psim_2| - \log|\Psim_1|  + \trace( \Psim_2^{-1} \Psim_1) - M.\label{berg_eq}
\end{align}
It is well-known that Bregman divergence $\scrD(\Psim_1, \Psim_2)$ is a measure of distance between  $\Psim_1$ and $\Psim_2$, although it is not a distance since it is not symmetric with respect to its arguments. Also, due to the strict convexity of $f(\Psim)=-\log |\Psim|$, we have  $\scrD(\Psim_1, \Psim_2) \geq 0$ with equality if and only if $\Psim_1=\Psim_2$. Moreover,  $\scrD(\Psim_1, \Psim_2)$ is convex with respect to its first argument $\Psim_1$ but not necessarily convex with respect to the second one $\Psim_2$. 
With this brief introduction, we can now illustrate the connection between our ML method and SPICE. One can see that, after dropping the constant terms, our proposed ML method is equivalent to finding a matrix $\Sigmam$ that minimizes the Bregman divergence  $\scrD( \widehat{\Sigmam}_\bfy, \Sigmam)$ of $\Sigmam$ with the sample covariance matrix $\widehat{\Sigmam}_\bfy$. By a little simplification, one can also show that the SPICE cost function in \eqref{eq:SPICE} can be written as $\trace(\Sigmam \widehat{\Sigmam}_{\bfy}^{-1}) + \trace(\Sigmam^{-1} \widehat{\Sigmam}_{\bfy})$. Therefore, from \eqref{berg_eq}, one can check that SPICE is equivalent to finding a covariance matrix $\Sigmam$ that minimizes the symmetric Bregman divergence $\scrD(\Sigmam, \widehat{\Sigmam}_\bfy)+\scrD(\widehat{\Sigmam}_\bfy, \Sigmam)$ with respect to  the sample covariance matrix $\widehat{\Sigmam}_\bfy$. 

Interestingly, compared with ML method that minimizes the non-convex Bregman divergence $\scrD( \widehat{\Sigmam}_\bfy, \Sigmam)$ (recall than $\scrD$ is generally non-convex with respect to its second argument), SPICE yields a convex optimization problem. 
It is generally known that, intuitively speaking, non-convex cost functions, provided that they can be properly minimized,  yield better performances  than the convex ones. In our case, we can see this evidently from the simulation results where our proposed non-convex ML method yields a better performance than  SPICE.  However, SPICE has the computational advantage that, due the convexity, it can be globally minimized. Although  generally we cannot guarantee the global optimality of the solution of our proposed CCCP method, our results indicate that the resulting solution is quite good in all range of parameters and simulation settings, and has always better performance than the estimate produced by SPICE. 

\subsection{ $\ell_{2,1}$-norm regularized Least Squares}
The second method we use for comparison is  $\ell_{2,1}$-norm regularization method proposed in \cite{haghighatshoar2018low}. Again assuming the dictionary $\Dm$ defined above, this method solves the following convex problem:
\begin{equation}\label{eq:mmv_formula}
\Wm^\star = \underset{\Wm \in \bC^{Q\times N}}{\arg \min}~ \frac{1}{2} \Vert\frac{1}{\sqrt{M}} \Dm\Wm -\Ym \Vert^2 + \sqrt{N} \Vert\Wm \Vert_{2,1}, 
\end{equation}
 where $\Vert\Wm \Vert_{2,1} = \sum_{i=1}^Q \Vert \Wm_{i,\cdot} \Vert$ denotes the $\ell_{2,1}$ norm. Then the covariance is estimated as $\Sigmam_{\hv}^{\ell_{2,1}} = \Dm \text{diag}(\uv^\star) \Dm^\herm$, where now the vector $\uv^\star$ consists of elements $u_i^\star =  \frac{\Vert \Wm_{i,\cdot} \Vert}{M \sqrt{N}}$.
 
For the comparisons we use the same dictionary matrix for all methods, except for the case in which we choose $\Gc_c$ to be a rectangular (non-delta) density family, since the rival methods are incompatible with the corresponding dictionary.

Denoting a generic covariance estimate as $\widetilde{\Sigmam}_\hv$, we use two error metrics to evaluate the estimation quality:\\
\begin{enumerate}
\item \textit{Normalized Frobenius-norm Error}: This error is defined as 
	\[ E_{\text{NF}} = \bE \left\{ \frac{\Vert \Sigmah -  \widetilde{\Sigmam}_\hv \Vert_\sfF}{\Vert \Sigmah \Vert_\sfF}\right\},  \]
	where the expectation is taken over random ASF realizations and random channel vector realizations given a specific ASF. 
	
	\item \textit{Grassmanian-distance Error}: This metric denotes the principal subspaces $\Um_J$ and $\tilde{\Um}_J$ of $\Sigmah$ and $\widetilde{\Sigmam}_\hv$ corresponding to their $J$ largest eigenvalues, where $J$ is the smallest integer satisfying $\frac{\sum_{i=1}^J \alpha_i}{\sum_{i=1}^M \alpha_i}>0.95$. Then the Grassmanian-distance Error between $\Sigmam_\hv$ and $\widetilde{\Sigmam}_\hv$ is defined as 
	\[ E_{\text{GD}} = \bE \left\{  \Vert \tauv \Vert_2\right\},  \]
	where $\tauv = [\tau_1,\ldots,\tau_J]^\transp$ is a $J$-dim vector such that $\cos (\tau_j),~j=1,\ldots,J$ are the eigenvalues of  $\bfU_J^\herm \widetilde{\bfU}_J$ \cite{miretti2018fdd}.	

	This metric shows how far the dominant subspaces of the true and estimated covariance matrices are from each other, which is an important factor in various applications of massive MIMO such as user grouping and group-based beamforming.
\end{enumerate}

\subsection{Setup for ULA and UPA}
In the simulations of this section, we consider a ULA with $M=20$ antennas and a UPA with $M = 5\times5=25$ antennas, where the spacing between two consecutive antenna elements is set to $d = \frac{\lambda}{2}$. Therefore, using \eqref{eq:a_vec}, the $i$-th element of the array response is given by $[\av (\xiv)]_i = e^{j \pi \langle \xiv , \rv_i \rangle}$ for $i=0,\ldots,M-1$. 
We produce random ASFs in the following general format:
\begin{equation}\label{eq:random_ASF}
\begin{aligned}
\gamma (\xiv) = \gamma_d (\xiv) + \gamma_c (\xiv) & = (\delta (\xiv - \xiv_1) + \delta (\xiv - \xiv_2))/4 \\
& + (\rect_{\Ac_1}+\rect_{\Ac_2})/2Z,
\end{aligned}
\end{equation}
where $\xiv_i, \forall i \in \{1, 2\}$ is chosen uniformly at random over $[-1,1]$ for the ULA and $||\xiv_i||_2^2 \leq 1$ for the UPA. Moreover, in the ULA $\Ac_1\subset [-1,0]$ and $\Ac_2\subset [0,1]$ and in UPA $\Ac_1$ and $\Ac_2$ are in different random quadrants of the unit circle. The width of $\Ac_1$  and $\Ac_2$ are chosen uniformly at random over $[0.1, 0.3]$ for the ULA and have an area uniformly chosen at random over $[0.3, 0.5]$ for the UPA. Furthermore, the normalization scalar $Z = \int_{\Omega} \rect_{\Ac_1}+\rect_{\Ac_2}\, d\xiv$, where in ULA $\Omega = [-1,1]$ and in UPA $\Omega$ is the unit circle. For every random ASF and for a fixed number of samples $N$, we generate 50 random, noisy channel sample matrices $\Ym = [\yv_0,\ldots,\yv_{N-1}]$ and estimate the channel covariance $\Sigmam_{\hv}$ using various methods. Then we average the estimation error over both the random sample realizations and 100 random ASFs. The SNR is set to $20$ dBs. The number of atoms in the dictionary of continuous densities $\Gc_c$ is $n=2M$ for the ULA and $n=9M$ for the UPA. The dictionary $\Gc_c$ used for parameterizing the continuous ASF component is chosen to be consisting of rectangular densities with non-overlapping support (see Fig. \ref{fig:dic_example}), i.e. $\psi_i (\xiv) = \rect_{\Ac_i}$ with $\Ac_i = [-1+\tfrac{2(i-1)}{n},-1+\tfrac{2i}{n}],~i=1,\ldots,n$ for the ULA and $\psi_{i,j} (\xiv) = \rect_{\Ac_i\times \Ac_j}$ with $\Ac_i = [-1+\tfrac{2(i-1)}{n},-1+\tfrac{2i}{n}], \Ac_j = [-1+\tfrac{2(j-1)}{n},-1+\tfrac{2j}{n}],~i,j=1,\ldots,\sqrt{n}$ for the UPA.  
\begin{figure*}[t]
	\centering
	\begin{subfigure}[b]{0.45\textwidth}
		\includegraphics[width=\textwidth]{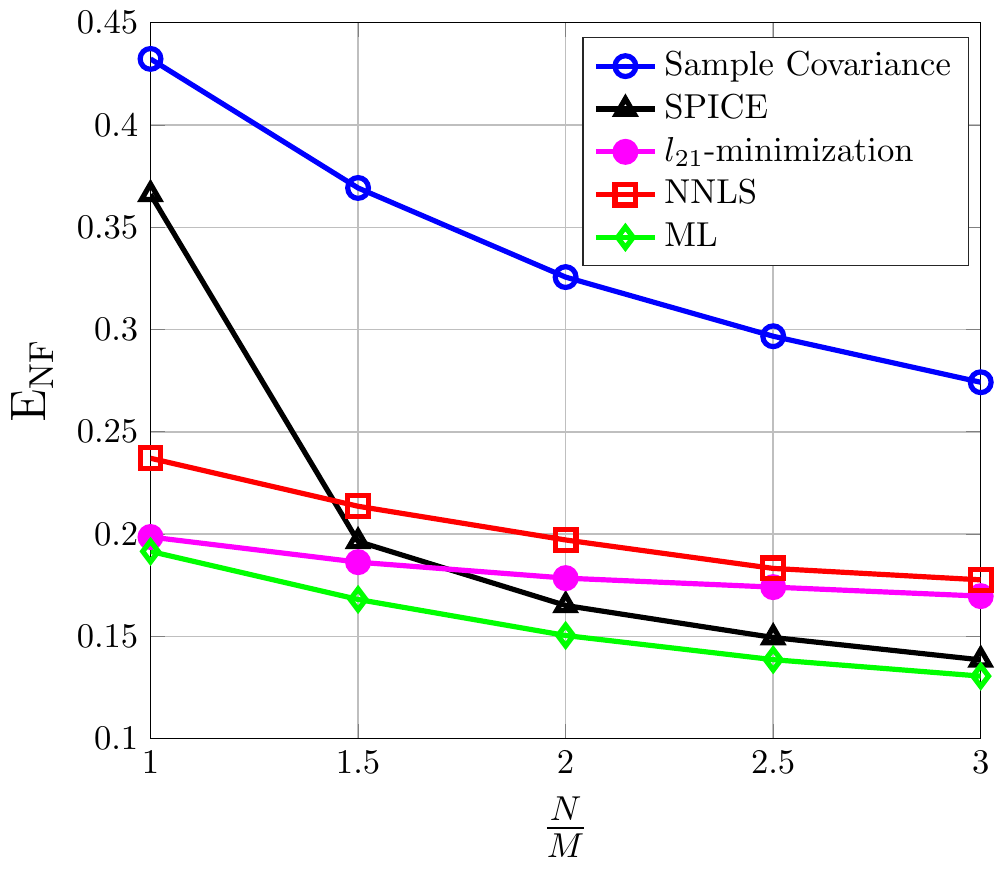}
		\caption{}
	\end{subfigure}
	~ 
	\begin{subfigure}[b]{0.45\textwidth}
		\includegraphics[width=\textwidth]{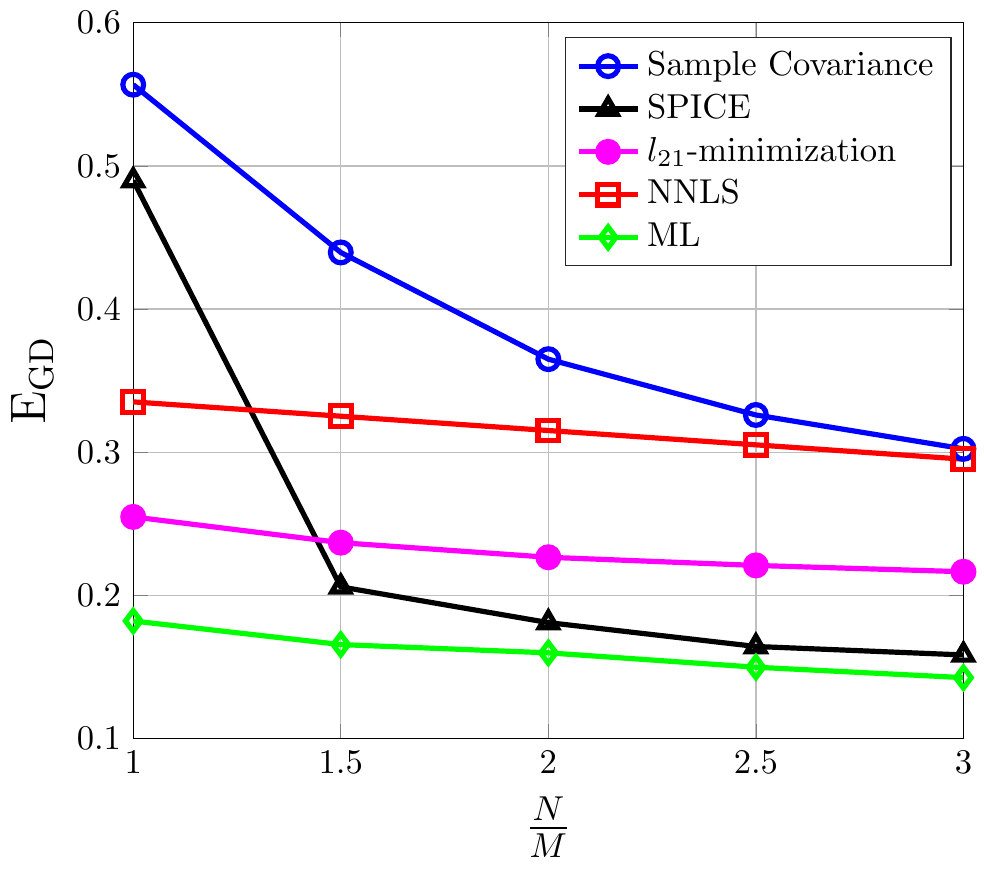}
		\caption{}
	\end{subfigure}
	
	\begin{subfigure}[b]{0.45\textwidth}
		\includegraphics[width=\textwidth]{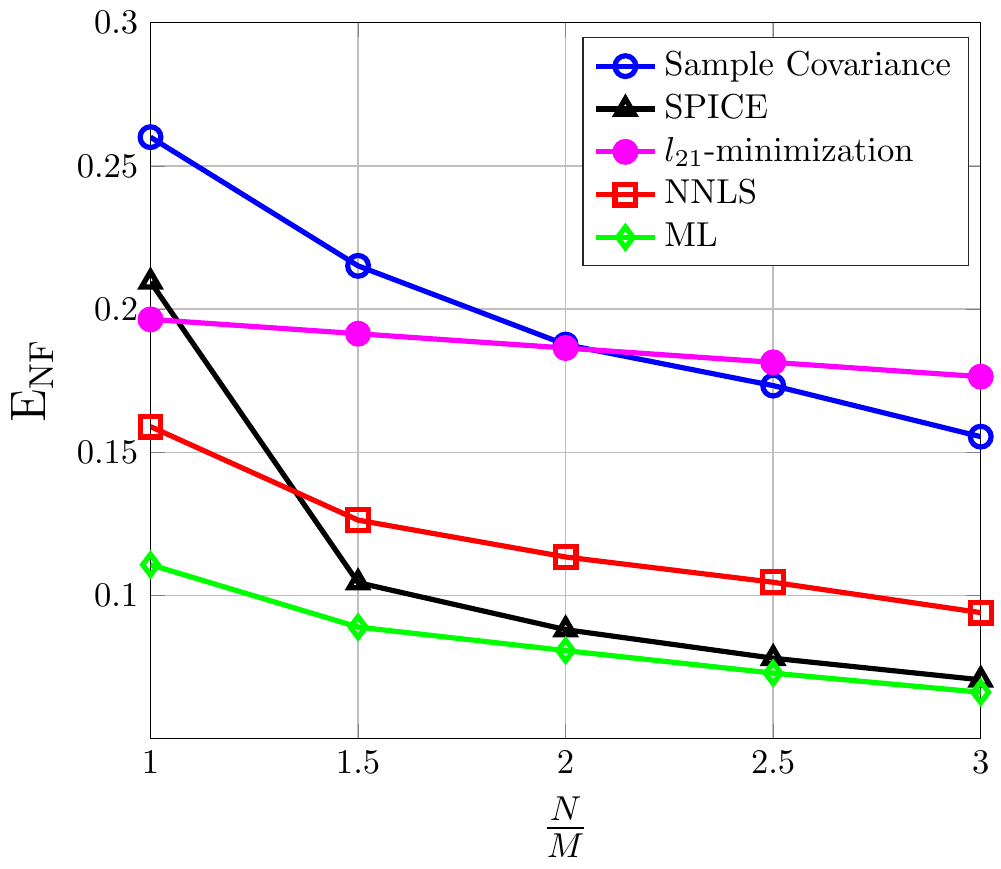}
		\caption{}
	\end{subfigure}
	~ 
	\begin{subfigure}[b]{0.45\textwidth}
		\includegraphics[width=\textwidth]{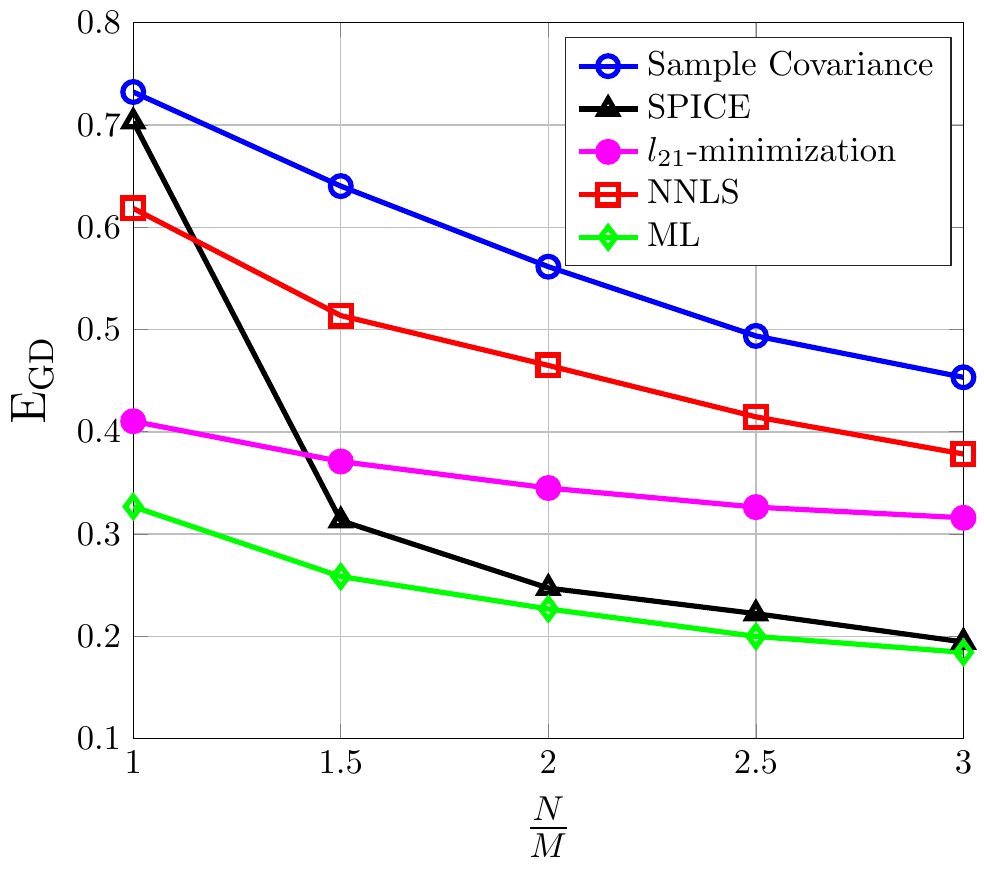}
		\caption{}
	\end{subfigure}
	\caption{Estimation quality comparison for a ULA with $M=20$ (upper row) and a UPA with $M=25$ (lower row): (a) and (c) normalized Frobenius-norm error, (b) and (d) Grassmanian-distance error. }\label{fig:UPA_Err}
\end{figure*} 
Note that this simulation setup is indeed in favor of the rival methods, in that none of them in their original form involves the design of an appropriate dictionary as we did using spike location estimation via MUSIC and introducing the dictionary of continuous kernels $\Gc_c$. Although this structured parameterization of the ASF comprised a major part of our proposed method, we decided to use it also for the rival methods to make the comparisons completely fair.

\subsection{Comparison for ULA}
In the upper row of Fig. \ref{fig:UPA_Err}, the comparison for ULA with normalized Frobenius-norm error and Grassmanian distance error are depicted. It is firstly observed that the results of the proposed ML method outperform the others for both metrics in all range of sample size, which indicates the advantage of our method. This is evident especially in small number of samples. The reason is that in such cases, the ML objective finds a much better fit to the available number of samples compared to the other methods. It is also observed that, the results of ML are better than the results of NNLS since we take the results of NNLS as the initial points of the proposed CCCP for ML algorithm. For example, although the results of SPICE are close to the proposed ML method under large sample size $N$, SPICE performs significantly worse than ML under small sample size $N$, e.g., $N/M=1$ for both metrics. As mentioned several times before, this is an important point, noting the particularly restrictive $N/M$ ratio in massive MIMO.
	
\subsection{Comparison for UPA} 
In the lower row of Fig. \ref{fig:UPA_Err}, the comparison for UPA with normalized Frobenius-norm error and Grassmanian distance error are depicted. Similar performances to the ULA case are observed for all methods. It is shown that under small number of sample size $N$ the proposed ML method performs much better than the other methods for both metrics. Again, by increasing the sample size $N$, the performance gap between all methods gets smaller as we enter the large-number-of-samples scenario.

\section{Conclusion}
We presented a massive MIMO channel covariance estimator. Using the specific structure of MIMO covariances, this estimator models a parametric representation of the ASF over the spatial domain and obtains the corresponding parameters by optimizing a maximum-likelihood objective. Our results show that the proposed method is superior to several state-of-the-art algorithms in the literature in terms of different performance metrics.
	
	\section{Appendix}
	\subsection{Theoretical Analysis of MUSIC and Further Discussion}\label{sec:perf_analysis}
	In this section, we provide a semi-rigorous analysis of the performance of the MUSIC algorithm for support estimation  for  a  ULA. Although the analysis does not extend verbatim to a 2D UPA, it gives insights on why MUSIC performs quite well for support estimation.
	The propositions and theorems are provided without proof. For the proofs, please see \cite{najim2016statistical}, which is the main source of the analysis provided in this section.

	We consider a ULA with $M$ elements and  with standard antenna spacing $\frac{\lambda}{2}$. We assume an ASF $\gamma(\xi)=\gamma_d(\xi) + \gamma_c(\xi)$ where $\gamma_d(\xi)$ is the ASF of the discrete part consisting of $r$ Dirac's delta functions as $\gamma_d(\xi)=\sum_{k=1}^r c_k \delta(\xi-\xi_k)$ and where $\gamma_c(\xi)$ is the ASF of the continuous part.
	In the scaling law studied in \cite{najim2016statistical}, the authors consider a scenario where the coefficient of spikes $\{c_k^{(M)}: k=1,\ldots,r\}$ scale with $M$ according to $c_k^{(M)}=\frac{\upsilon_k}{M}$ where $\{\upsilon_k:k=1,\ldots,r\}$ are positive constants encoding the relative strength of the spikes. As a result, the scaling regime in \cite{najim2016statistical} covers a more challenging scenario where the amplitudes of the spikes decrease by increasing $M$ such that identifying them becomes more and more challenging.
	In this section, we will first state the results that follow from the rigorous analysis of this more challenging scaling regime in \cite{najim2016statistical}. 
	In wireless applications, in contrast, the spike coefficients $\{c_k: k=1,\ldots,r\}$ remain the same, regardless of the number of BS antennas $M$. We adjust to this setup by assuming that the coefficients $\upsilon_k$ are also growing proportionally to $M$ like $\upsilon_k = M c_k$ and then use the results in \cite{najim2016statistical}. This is the only point that makes our analysis in this part semi-rigorous.
	
	In the following, we first focus on the scaling regime in \cite{najim2016statistical}.
	We consider a generic channel vector $\bfh$ with a covariance matrix
	\begin{equation}\label{eq:channel_cov_an}
	\Sigmah =  \Sigmah^d + \Sigmah^c :=  \sum_{k=1}^{r} c_k^{(M)} \av (\xi_k) \av (\xi_k)^\herm + \int_{-1}^1 \gamma_c (\xi) \av (\xi) \av (\xi)^\herm d \xi,
	\end{equation}
	as in \eqref{eq:channel_cov} where $\bfa(\xi)=(1, e^{j \pi \xi}, \dots, e^{j\pi (M-1)\xi})^\transp$ denotes the array response vector of the ULA. We denote the set of noisy samples by $\{\bfy(s)=\bfh(s)+ \bfz(s): s=1,\ldots,N\}$ and their covariance matrices by $\Sigmam_\bfy=\Sigmam_\bfh + N_0 \bfI$, where $N_0$ is the noise variance.
	Let $\lambda_{1,M}\geq \ldots \geq \lambda_{M,M}$ denote the singular values of  $\Sigmam_\bfy$. Note that in this section we consider an asymptotic analysis where $M$ approaches infinity, thus, we use the notation $\lambda_{k,M}$ for $k=1,\ldots,M$, to illustrate the explicit dependence of the $M$ singular values of $\Sigmam_\bfy$ on $M$. The following proposition shows that as $M$ grows, the $r$ largest singular values of $\Sigmam_\bfy$ ``escape" from the rest of the singular values and converge to fixed values as $M\to \infty$. 
	\begin{proposition}[Escape of the $r$ largest singular values of the covariance matrix $\Sigmam_\bfy$]\label{prop:sep_asymp}
		Consider $\gamma(\xi)$ as before with a discrete part $\gamma_d(\xi)=\sum_{k=1}^r c_k^{(M)} \delta(\xi-\xi_k)$ with $c_k^{(M)}=\frac{\upsilon_k}{M}$ and a continuous part $\gamma_c(\xi)$, and let $\Sigmam_\bfh$ be the covariance matrix generated by $\gamma(\xi)$ as in \eqref{eq:channel_cov_an} and let $\Sigmam_\bfy=\Sigmam_\bfh + N_0 \bfI$. 
		If $\upsilon_k+\gamma_c (\xi_k) > \Vert \gamma_c\Vert_\infty $ for all $k=1,\ldots,r,$ then 
		\begin{align}
		\lambda_{k,M}\underset{M\to \infty}{\longrightarrow} \lambda_k =\upsilon_k+\gamma_c (\xi_k)+N_0, \label{szego_res}
		\end{align}
		while $\lim \sup_{M\to \infty} \lambda_{r+1,M} \le \Vert \gamma_c \Vert_\infty + N_0$ (note that $\|\gamma_c\|_\infty:=\sup_{\xi \in [-1,1]} \gamma_c(\xi)$). \hfill $\square$
	\end{proposition}
	
	This proposition can be interpreted as follows: if for all $k=1,\ldots,r$, the coefficient of the spike plus the value of the continuous  component $\gamma_c(\xi)$ at the location of the spike is greater than the supremum of the continuous part $\gamma_c$ over the whole set of AoAs (that is $\Vert \gamma_c\Vert_\infty $), then the $r$ largest singular values converge  according to \eqref{szego_res} to a value larger than $\|\gamma_c\|_\infty + N_0$ (due to the assumption of the proposition $\upsilon_k + \gamma_c(\xi_k) > \|\gamma_c\|_\infty$) while the rest of the singular values are upper-bounded by $\Vert \gamma_c\Vert_\infty + N_0$. 
	As a result, by increasing the dimension $M$, we can evidence a nice separation between the first $r$ singular values and the remaining $M-r$ ones, which can exploited to identify the number of spikes $r$. 
	
	Unfortunately, Proposition \ref{prop:sep_asymp} is not directly applicable in our case since we have access only to the sample covariance of the noisy samples, namely, $\wh{\Sigmam}_\bfy=\frac{1}{N} \sum_{s=1}^N \bfy(s) \bfy(s)^\herm$, rather than their true covariance matrix $\Sigmam_\bfy$. 
	Fortunately, this result can be modified to work also for $\wh{\Sigmam}_\bfy$ provided that the number of available signal samples for covariance estimation $N$ is sufficiently large. To characterize this rigorously, we consider an asymptotic regime where the number of samples $N$ grows proportionally to the number of antennas $M$  such that $\zeta_M:=\frac{M}{N} \to \zeta>0$ as $M\to \infty$. Of course, in practice $M$ is always limited but this asymptotic scaling law gives a flavor of conditions under which the support recovery is feasible for the discrete spikes. 
	
	Let $\wh{\lambda}_{1,M}, \dots, \wh{\lambda}_{M,M}$ be the singular values of $\wh{\Sigmam}_\bfy$ 
	and let us define the empirical distribution of these singular values by 
	\[ \widehat{\gamma}_M(\lambda) = \frac{1}{M} \sum_{k=1}^M \delta (\lambda - \widehat{\lambda}_{k,M}). \]
	Then one can show that  \cite{baik2006eigenvalues} almost surely (a.s.) as $M\to \infty$
	\[ \widehat{\gamma}_M  \overset{\text{weakly}}{\longrightarrow} \gamma, \]
	where $\gamma$ is a deterministic density  characterized by its Stieltjes transform as
	\begin{align}
	\varepsilon(z) = \int_{\bR} \frac{d\gamma (\lambda)}{\lambda - z}, \label{eps_relation1}
	\end{align}
	where $\varepsilon(z)$ is a function that satisfies the following fixed-point equation 
	\begin{align}
	\varepsilon(z) = \int_{\bR} \frac{d\nu (\lambda)}{\lambda (1-\zeta -\zeta\, z\, \varepsilon(z))-z}, \label{eps_relation2}
	\end{align}
	for all $z\in \bC \backslash\text{supp} (\gamma)$, where $\text{supp} (\gamma)$ denotes the support of the distribution $\gamma$. In  \eqref{eps_relation2},  $\nu (\lambda)$ is the asymptotic distribution of the singular values of the true covariance $\Sigmam_\bfy=\Sigmah+ N_0 \bfI$, namely, $\nu_M = \frac{1}{M} \sum_{k=1}^M \delta (\lambda - \lambda_{k,M})\overset{\text{weakly}}{\longrightarrow} \nu$, where $\{\lambda_{k,M}: k \in [M]\}$ denotes the set of singular values of $\Sigmam_\bfy=\Sigmam_\bfh + N_0 \bfI$ as before. It is worthwhile to mention that although the singular values of $\Sigmam_\bfy$ and $\widehat{\Sigmam}_\bfy$ have a well-defined limit as $\gamma(\lambda)$ and $\nu(\lambda)$, these two limit distributions are different from each other for any $\zeta>0$ and approach each other as $\zeta \to 0$, namely, when the number of samples $N$ becomes tremendously larger than $M$, where in that case $\widehat{\Sigmam}_\bfy$ also converges to $\Sigmam_\bfy$.
	
	To extend the separation condition proved in  Proposition \ref{prop:sep_asymp} for the true covariance $\Sigmam_\bfy$ to the sample covariance $\widehat{\Sigmam}_\bfy$, we need to study $\nu(\lambda)$ further.
	From Szeg\"o's theorem \cite{grenander1958toeplitz}, it is
	well-known that $\nu$ is given by the distribution of the random variable $\gamma_c(\bar{\xi}) + N_0$ when $\bar{\xi}\sim \clU([-1,1])$ is uniformly distributed  in $[-1,1]$. Note that since the random variable $\gamma_c(\bar{\xi}) + N_0$ is upper bounded by $\|\gamma_c\|_\infty + N_0$, the support of the distribution $\nu$ lies always in the interval $\big[0, \|\gamma_c\|_\infty +N_0\big ]$, and in particular $\max(\supp(\nu)) =\|\gamma_c\|_\infty + N_0$. This  implies that the function $\omega \mapsto \phi(\omega)$ defined by 
	\begin{align}
	\phi (\omega) = \omega \left( 1- \zeta\int_{\bR} \frac{\lambda}{\lambda - \omega } d\nu (\lambda) \right), \label{phi_func}
	\end{align}
	is well-defined for all $\omega \in (\|\gamma_c\|_\infty + N_0, +\infty)$. Note that $\phi(\omega)$ is a continuous and differentiable (of any order) function  in this interval. Moreover, $\phi(\omega) \to \infty$ as $\omega \to \|\gamma_c\|_\infty + N_0$ from the right and $\lim _{\omega \to \infty} \phi(\omega)=\infty$. Thus, $\phi(\omega)$ should have a local minimum $\omega_0 \in (\|\gamma_c\|_\infty + N_0, \infty)$. A direct computation  shows that 
	\begin{align}
	\phi^{''}(\omega)= \int \frac{2\zeta \lambda^2}{(\omega-\lambda)^3} d\nu(\lambda),
	\end{align}
	which is alway positive in the interval $(\|\gamma_c\|_\infty+N_0, \infty)$. Hence, $\phi(\omega)$ is a convex function in this interval and $\omega_0$ is the unique minimizer of $\phi(\omega)$. It would be also interesting to investigate the dependence of $\omega_0$ on the asymptotic sampling ratio $\zeta$. Note that $\omega_0$ is the unique minimizer of $\phi(\omega)$, thus, it satisfies $\phi'(\omega_0)=0$. Taking the derivative of $\phi(\omega)$, we can write the condition $\phi'(\omega_0)=0$ as
	\begin{align}
	\int \frac{\lambda^2}{(\omega_0-\lambda)^2} d\nu(\lambda) = \frac{1}{\zeta}.
	\end{align}
	We can simply check that $\omega_0(\zeta)$ is an increasing function of $\zeta$. In particular, by increasing the number of samples $N$, thus, letting $\zeta \to 0$, we have $\frac{1}{\zeta} \to \infty$, which is satisfied provided that $\omega_0$ approaches the boundary value $\|\gamma_c\|_\infty + N_0$. Similarly, we can check that by decreasing the number of samples $N$ in a scaling regime where $\zeta \to \infty$, we obtain $\omega_0=\infty$. In brief, $\omega_0$ ranges monotonically in the interval $(\|\gamma_c\|_\infty + N_0, +\infty)$ for $\zeta \in (0, +\infty)$. 
	
	The following theorem shows that, similar to the escape of the $r$ largest singular values in the spectrum of the true covariance $\Sigmam_\bfy$ illustrated in Proposition \ref{prop:sep_asymp}, the $r$ largest singular values of the sample covariance $\widehat{\Sigmam}_\bfy$ escape from the rest of its spectrum if a ``separation" condition is satisfied. This separation condition can be formulated in terms of $\omega_0$ as follows.
	\begin{theorem}\label{thm:spike_num}
		Let $\{\lambda_{k,M}: k =1,\ldots,M\}$ denote the set of singular values of $\Sigmam_\bfy$ as before and suppose that
		\begin{equation}\label{eq:sep_condition}
		\lambda_{r,M} > \omega_0. 
		\end{equation}
		Then, for $k=1,\ldots,r$, with probability one as $M\to \infty$ we have
		\begin{equation}
		\widehat{\lambda}_{k,M} \to \phi (\lambda_k),
		\end{equation}
		whereas $\widehat{\lambda}_{r+1,M} \to \phi (\omega_0)<\phi (\lambda_{r,M})$ (as $\omega_0$ is the maximizer of $\phi(\omega)$). \hfill $\square$
	\end{theorem}
	
	The separation condition $\eqref{eq:sep_condition}$ implicitly depends on the parameters of the spike elements $\{\upsilon_k: k =1,\ldots,r\}$ as well as the continuous part of the ASF $\gamma_c$ (through the function $\phi(\omega)$ defined in \eqref{phi_func}) and in particular on the asymptotic sampling ratio $\zeta$.
	As a sanity check by increasing the number of samples $\zeta \to 0$, and $\omega_0(\zeta) \to \|\gamma_c\|_\infty + N_0$, and the separation condition in Theorem \ref{thm:spike_num} becomes the same as that in Proposition \ref{prop:sep_asymp}, which makes sense since for large number of samples the sample covariance matrix $\wh{\Sigmam}_\bfy$ converges to the original covariance matrix $\Sigmam_\bfy$. Moreover, as $\omega_0(\zeta) > \|\gamma_c\|_\infty + N_0$ for all $\zeta$, Theorem \ref{thm:spike_num} requires a stronger separation condition than Proposition \ref{prop:sep_asymp}, which is the cost one needs to pay for not having the original  covariance matrix but the sample covariance matrix. 
	
	Overall, if the separation condition \eqref{eq:sep_condition} is satisfied, we are able to consistently detect the number of spikes by identifying the gap between the singular values. In particular, 
	\[ \wh{r}_M = \max \{ k: \widehat{\lambda}_{k,M} >\phi (\omega_0)+\epsilon  \}\overset{a.s.}{\longrightarrow} r, \]
	as $M\to \infty$ for  any $\epsilon \in \big(0, \phi (\lambda_{r,M})-\phi (\omega_0)\big)$.
	
	Eventually, the following theorem proves the consistency of the MUSIC estimator in the present context. 
	\begin{theorem}[Consistency of MUSIC]
		If the separation condition \eqref{eq:sep_condition} holds, then
		\begin{equation}
		M\, (  \wh{\xi}_k - \xi_k  ) \overset{a.s.}{\longrightarrow} 0, \ k=1,\dots, r,
		\end{equation}
		as $M\to \infty$, where $ \wh{\xi}_k$ denotes the AoA estimate of  the $k$-th dominant minima of \eqref {eq:pseudo_spectrum_1}. \hfill $\square$
	\end{theorem}

	\vspace{2mm}
	
  In brief, the theoretical analysis presented in this section, illustrates that in the (challenging) scaling regime where the coefficients of the spikes vary as $c_k^{(M)}=\frac{\upsilon_k}{M}$, thus, decrease by increasing $M$, and for any  finite fraction $\zeta_M=\frac{M}{N} \to \zeta \in (0, \infty)$, one is able to estimate consistently the number and the support of the spikes through MUSIC algorithm provided that the weights of the parameters of the  spikes $\{\upsilon_k: k =1,\ldots,r\}$ are sufficiently large such that they stand out of the amplitude of the continuous part $\gamma_c$. The degree up to which these weights should be large depends on $\zeta$, where in the best case of very large number of samples, $N\gg M$ such that $\zeta \approx 0$, one needs at least $\upsilon_k+ \gamma_c(\xi_k) \geq \|\gamma_c\|_\infty$. In general, for all other values of $\zeta$, the separability condition is satisfied provided that the condition in \eqref{eq:sep_condition} is fulfilled where again, intuitively speaking, the weights $\{\upsilon_k: k =1,\ldots,r\}$ should be large enough to make sure that the first $r$ singular values of $\Sigmam_\bfy$ (which  of course grow by increasing the weights $\{\upsilon_k: k =1,\ldots,r\}$) pass the threshold $\omega_0$ illustrated in \eqref{eq:sep_condition}.
	
	It is worthwhile here to pose these results in the semi-rigorous setting we already discussed. More specifically, since in our case the amplitude of the spikes  $\{c_k: k =1,\ldots,r\}$ remain constant (rather than decreasing with $M$), we can follow the same reasoning by assuming that the  coefficients $\upsilon_k$ grow proportional to $M$ as $\upsilon_k=M c_k$. 
	
	An important point of analysis in  \cite{najim2016statistical} summarized in this section is that for any   asymptotic sampling ratio $\frac{M}{N} \to \zeta \in (0, \infty)$, no matter how small $\zeta$ may be, we can make the detection of all $r$ spikes, namely, their number $r$ and also their support, feasible by increasing the coefficients $\{\upsilon_k: k =1,\ldots,r\}$ until the separability condition in \eqref{eq:sep_condition} is fulfilled. Let us first illustrate this point step by step. First note that the measure $\nu(\lambda)$ depends only on the continuous part $\gamma_c$, thus, is not affected by changing the weights $\{\upsilon_k: k =1,\ldots,r\}$. Therefore, for a fixed $\zeta\in (0, \infty)$, the function $\phi(\omega)$ and as a result the parameter $\omega_0$ are not affected by changing $\{\upsilon_k: k =1,\ldots,r\}$. 
	Second, by dropping the contribution of the continuous part $\gamma_c$ from $\Sigmam_\bfy$, we can easily check that the first $r$ singular values of $\Sigmam_\bfy$ are larger than the first $r$ singular values of the matrix 
	\begin{align}
	\Sigmam_\bfh^d + N_0\bfI =\sum_{k=1}^r \frac{\upsilon_k}{M} \bfb(\xi_k)  \bfb(\xi_k)^\herm + N_0\bfI.
	\end{align}
	A direct calculation shows that as $M \to \infty$, and as a result $\frac{\inp{ \bfb(\xi_k)}{ \bfb(\xi_{k'})}}{M} \to 0$ for $k \not = k'$,  the $r$ largest singular values  approach $\{\upsilon_k + N_0: k =1,\ldots,r\}$ which would satisfy \eqref{eq:sep_condition} by increasing $\{\upsilon_k: k =1,\ldots,r\}$ (as $\omega_0$ is not affected by changing $\{\upsilon_k: k =1,\ldots,r\}$). In brief, we can say that the separability condition \eqref{eq:sep_condition} is fulfilled if
	\begin{align}
	\min \{ \upsilon_k: k =1,\ldots,r\} \geq \eta(\zeta, \gamma_c),
	\end{align}
	where $\eta(\zeta, \gamma_c)$ is a finite threshold that depends on the sampling ratio $\zeta$ and $\gamma_c$. As a result, assuming that $\upsilon_k=M\, c_k $ grows proportionally to $M$, this condition would be satisfied for any finite $\zeta \in (0, \infty)$ and for any practically relevant $\gamma_c$ provided that  $M$ is sufficiently large. 
	
	We would like to emphasize that the importance of this (semi-rigorous) result is that, it implies that no matter how small the spike amplitudes $\{c_k: k =1,\ldots,r\}$ are and no matter how small the number of  samples $N$ is compared with $M$ (of course provided that the asymptotic sampling ratio $\zeta$ remains finite), MUSIC algorithm would be able to recover all the spikes if $M$ is sufficiently large. Looking from another perspective, of course we know that if we do not have enough number of samples $N$ (namely if $\zeta$ is quite large but finite), it may be impossible information-theoretically to estimate some generic covariance matrices. However, even in those scenarios, MUSIC algorithm would be asymptotically consistent. This provides a strong guarantee that  MUSIC algorithm works perfectly without incurring any sample complexity for the covariance estimation.

	{\small
		\bibliographystyle{IEEEtran}
		\bibliography{references}
	}
	
\end{document}

%% file: symbols.tex
\usepackage{mathbbol}

\def\mindex#1{\index{#1}}

\def\sq{\hbox{\rlap{$\sqcap$}$\sqcup$}}
\def\qed{\ifmmode\sq\else{\unskip\nobreak\hfil
\penalty50\hskip1em\null\nobreak\hfil\sq
\parfillskip=0pt\finalhyphendemerits=0\endgraf}\fi\medskip}

\long\def\defbox#1{\framebox[.9\hsize][c]{\parbox{.85\hsize}{%
\parindent=0pt
\baselineskip=12pt plus .1pt      
\parskip=6pt plus 1.5pt minus 1pt 
 #1}}}

\long\def\beginbox#1\endbox{\subsection*{}%
\hbox{\hspace{.05\hsize}\defbox{\medskip#1\bigskip}}%
\subsection*{}}

\def\endbox{}

\def\diag{{\text{diag}}}

\def\tr{\mathsf{tr}}

\def\supp{{\rm supp\,}}

\newsavebox{\junk}
\savebox{\junk}[1.6mm]{\hbox{$|\!|\!|$}}

\def\det{{\mathop{\rm det}}}

\def\argmin{\mathop{\rm arg\, min}}

\def\bC{{\mathbb C}}

\def\bE{{\mathbb E}}

\def\bR{{\mathbb R}}

\def\bfA{{\bf A}}

\def\bfI{{\bf I}}

\def\bfS{{\bf S}}

\def\bfU{{\bf U}}

\def\bfY{{\bf Y}}

\def\bfa{{\bf a}}
\def\bfb{{\bf b}}

\def\bff{{\bf f}}

\def\bfh{{\bf h}}

\def\bfu{{\bf u}}

\def\bfy{{\bf y}}
\def\bfz{{\bf z}}

\def\scrC{{\mathscr{C}}}
\def\scrD{{\mathscr{D}}}

\def\ttc{{\mathtt c}}

\def\tte{{\mathtt e}}

\def\ttr{{\mathtt r}}

\def\ttt{{\mathtt t}}

\def\sfF{{\sf F}}

\def\sfH{{\sf H}}

\def\bfmath#1{{\mathchoice{\mbox{\boldmath$#1$}}%
{\mbox{\boldmath$#1$}}%
{\mbox{\boldmath$\scriptstyle#1$}}%
{\mbox{\boldmath$\scriptscriptstyle#1$}}}}

\def\bfmY{\bfmath{Y}}

\def\bfmhhaY{\bfmath{\hhaY}} 
\def\bfmhhaY{\hbox to 0pt{$\widehat{\bfmY}$\hss}\widehat{\phantom{\raise 1.25pt\hbox{$\bfmY$}}}}

\def\til={{\widetilde =}}

\def\clC{{\cal C}}

\def\clG{{\cal G}}

\def\clM{{\cal M}}
\def\clN{{\cal N}}

\def\clU{{\cal U}}

 \def\FRAC#1#2#3{\genfrac{}{}{}{#1}{#2}{#3}}

\def\ddtp{{\mathchoice{\FRAC{1}{d^{\hbox to 2pt{\rm\tiny +\hss}}}{dt}}%
{\FRAC{1}{d^{\hbox to 2pt{\rm\tiny +\hss}}}{dt}}%
{\FRAC{3}{d^{\hbox to 2pt{\rm\tiny +\hss}}}{dt}}%
{\FRAC{3}{d^{\hbox to 2pt{\rm\tiny +\hss}}}{dt}}}}

\def\average#1,#2,{{1\over #2} \sum_{#1}^{#2}}

\def\eye(#1){{\bf(#1)}\quad}

\newtheorem{theorem}{{\bf Theorem}}
\newtheorem{proposition}{{\bf Proposition}}

\def\eq#1/{(\ref{e:#1})}

\newcommand{\inp}[2]{{\langle #1, #2 \rangle}}

\newcommand{\beqn}[1]{\notes{#1}%
\begin{eqnarray} \elabel{#1}}

\newcommand{\eeqn}{\end{eqnarray} }

\newcommand{\beq}[1]{\notes{#1}%
\begin{equation}\elabel{#1}}

\newcommand{\eeq}{\end{equation}}

\def\bdes{\begin{description}}
\def\edes{\end{description}}

\newcounter{rmnum}

\newcounter{anum}

{\end{list}}

\def\ass(#1:#2){(#1\ref{#1:#2})}

\def\ritem#1{
\item[{\sf \ass(\current_model:#1)}]
}

\newenvironment{recall-ass}[1]{%
\begin{description}
\def\current_model{#1}}{
\end{description}
}



%% file: macros.tex
\long\def\comment#1{}

\newcommand{\av}{{\bf a}}
\newcommand{\bv}{{\bf b}}
\newcommand{\cv}{{\bf c}}

\newcommand{\hv}{{\bf h}}

\newcommand{\rv}{{\bf r}}

\newcommand{\uv}{{\bf u}}

\newcommand{\yv}{{\bf y}}
\newcommand{\zv}{{\bf z}}

\newcommand{\Dm}{{\bf D}}

\newcommand{\Sm}{{\bf S}}

\newcommand{\Um}{{\bf U}}
\newcommand{\Wm}{{\bf W}}

\newcommand{\Ym}{{\bf Y}}

\newcommand{\Ac}{{\cal A}}

\newcommand{\Gc}{{\cal G}}

\newcommand{\tauv}{\hbox{\boldmath$\tau$}}

\newcommand{\xiv}{\hbox{\boldmath$\xi$}}

\newcommand{\Gammam}{\hbox{\boldmath$\Gamma$}}
\newcommand{\Lambdam}{\hbox{\boldmath$\Lambda$}}

\newcommand{\Sigmam}{\hbox{\boldmath$\Sigma$}}

\newcommand{\Psim}{\hbox{\boldmath$\Psi$}}

\newcommand{\Xim}{\hbox{\boldmath$\Xi$}}

\renewcommand{\det}{{\hbox{det}}}
\newcommand{\trace}{{\hbox{tr}}}
\renewcommand{\arg}{{\hbox{arg}}}

\newcommand{\transp}{{\sf T}}
\renewcommand{\vec}{{\rm vec}}



%% file: paper.bbl
\begin{thebibliography}{10}
\providecommand{\url}[1]{#1}
\csname url@samestyle\endcsname
\providecommand{\newblock}{\relax}
\providecommand{\bibinfo}[2]{#2}
\providecommand{\BIBentrySTDinterwordspacing}{\spaceskip=0pt\relax}
\providecommand{\BIBentryALTinterwordstretchfactor}{4}
\providecommand{\BIBentryALTinterwordspacing}{\spaceskip=\fontdimen2\font plus
\BIBentryALTinterwordstretchfactor\fontdimen3\font minus
  \fontdimen4\font\relax}
\providecommand{\BIBforeignlanguage}[2]{{%
\expandafter\ifx\csname l@#1\endcsname\relax
\typeout{** WARNING: IEEEtran.bst: No hyphenation pattern has been}%
\typeout{** loaded for the language `#1'. Using the pattern for}%
\typeout{** the default language instead.}%
\else
\language=\csname l@#1\endcsname
\fi
#2}}
\providecommand{\BIBdecl}{\relax}
\BIBdecl

\bibitem{khalilsarai2018fdd}
M.~B. Khalilsarai, S.~Haghighatshoar, X.~Yi, and G.~Caire, ``{FDD} massive
  {MIMO} via {UL}/{DL} channel covariance extrapolation and active channel
  sparsification,'' \emph{IEEE Transactions on Wireless Communications},
  vol.~18, no.~1, pp. 121--135, 2018.

\bibitem{boroujerdi2018low}
M.~N. Boroujerdi, S.~Haghighatshoar, and G.~Caire, ``Low-complexity
  statistically robust precoder/detector computation for massive {MIMO}
  systems,'' \emph{IEEE Transactions on Wireless Communications}, vol.~17,
  no.~10, pp. 6516--6530, 2018.

\bibitem{haghighatshoar2018low}
S.~Haghighatshoar and G.~Caire, ``Low-complexity massive {MIMO} subspace
  estimation and tracking from low-dimensional projections,'' \emph{IEEE
  Transactions on Signal Processing}, vol.~66, no.~7, pp. 1832--1844, 2018.

\bibitem{haghighatshoar2017massive}
------, ``Massive {MIMO} channel subspace estimation from low-dimensional
  projections,'' \emph{IEEE Trans. on Signal Processing}, vol.~65, no.~2, pp.
  303--318, 2017.

\bibitem{adhikary2013joint}
A.~Adhikary, J.~Nam, J.-Y. Ahn, and G.~Caire, ``Joint spatial division and
  multiplexing: the large-scale array regime,'' \emph{{IEEE Trans. on Inform.
  Theory}}, vol.~59, no.~10, pp. 6441--6463, 2013.

\bibitem{marchenko1967distribution}
V.~A. Marchenko and L.~A. Pastur, ``Distribution of eigenvalues for some sets
  of random matrices,'' \emph{Matematicheskii Sbornik}, vol. 114, no.~4, pp.
  507--536, 1967.

\bibitem{hachem2005empirical}
W.~Hachem, P.~Loubaton, and J.~Najim, ``The empirical eigenvalue distribution
  of a {G}ram matrix: From independence to stationarity,'' \emph{arXiv preprint
  math/0502535}, 2005.

\bibitem{couillet2011random}
R.~Couillet and M.~Debbah, \emph{Random matrix methods for wireless
  communications}.\hskip 1em plus 0.5em minus 0.4em\relax Cambridge University
  Press, 2011.

\bibitem{pourahmadi2013high}
M.~Pourahmadi, \emph{High-dimensional covariance estimation: with
  high-dimensional data}.\hskip 1em plus 0.5em minus 0.4em\relax John Wiley \&
  Sons, 2013, vol. 882.

\bibitem{ravikumar2011high}
P.~Ravikumar, M.~J. Wainwright, G.~Raskutti, B.~Yu \emph{et~al.},
  ``High-dimensional covariance estimation by minimizing l1-penalized
  log-determinant divergence,'' \emph{Electronic Journal of Statistics},
  vol.~5, pp. 935--980, 2011.

\bibitem{chen2011robust}
Y.~Chen, A.~Wiesel, and A.~O. Hero, ``Robust shrinkage estimation of
  high-dimensional covariance matrices,'' \emph{IEEE Transactions on Signal
  Processing}, vol.~59, no.~9, pp. 4097--4107, 2011.

\bibitem{friedman2008sparse}
J.~Friedman, T.~Hastie, and R.~Tibshirani, ``Sparse inverse covariance
  estimation with the graphical lasso,'' \emph{Biostatistics}, vol.~9, no.~3,
  pp. 432--441, 2008.

\bibitem{stoica2005spectral}
P.~Stoica, R.~L. Moses \emph{et~al.}, ``Spectral analysis of signals,'' 2005.

\bibitem{najim2016statistical}
O.~Najim, P.~Vallet, G.~Ferr{\'e}, and X.~Mestre, ``On the statistical
  performance of {MUSIC} for distributed sources,'' in \emph{2016 IEEE
  Statistical Signal Processing Workshop (SSP)}.\hskip 1em plus 0.5em minus
  0.4em\relax IEEE, 2016, pp. 1--5.

\bibitem{tse2005fundamentals}
D.~Tse and P.~Viswanath, \emph{Fundamentals of wireless communication}.\hskip
  1em plus 0.5em minus 0.4em\relax Cambridge university press, 2005.

\bibitem{yuille2003concave}
A.~L. Yuille and A.~Rangarajan, ``The concave-convex procedure,'' \emph{Neural
  computation}, vol.~15, no.~4, pp. 915--936, 2003.

\bibitem{schmidt1986multiple}
R.~O. Schmidt, ``Multiple emitter location and signal parameter estimation,''
  \emph{Antennas and Propagation, IEEE Transactions on}, vol.~34, no.~3, pp.
  276--280, 1986.

\bibitem{stoica1989music}
P.~Stoica and A.~Nehorai, ``{MUSIC}, maximum likelihood, and {C}ramer-{R}ao
  bound,'' \emph{IEEE Transactions on Acoustics, speech, and signal
  processing}, vol.~37, no.~5, pp. 720--741, 1989.

\bibitem{mestre2008improved}
X.~Mestre, ``Improved estimation of eigenvalues and eigenvectors of covariance
  matrices using their sample estimates,'' \emph{IEEE Transactions on
  Information Theory}, vol.~54, no.~11, pp. 5113--5129, 2008.

\bibitem{stoica2011spice}
P.~Stoica, P.~Babu, and J.~Li, ``Spice: A sparse covariance-based estimation
  method for array processing,'' \emph{IEEE Transactions on Signal Processing},
  vol.~59, no.~2, pp. 629--638, 2011.

\bibitem{bregman1967relaxation}
L.~M. Bregman, ``The relaxation method of finding the common point of convex
  sets and its application to the solution of problems in convex programming,''
  \emph{USSR computational mathematics and mathematical physics}, vol.~7,
  no.~3, pp. 200--217, 1967.

\bibitem{miretti2018fdd}
L.~Miretti, R.~L. Cavalcante, and S.~Stanczak, ``{FDD} massive {MIMO} channel
  spatial covariance conversion using projection methods,'' in \emph{2018 IEEE
  International Conference on Acoustics, Speech and Signal Processing
  (ICASSP)}.\hskip 1em plus 0.5em minus 0.4em\relax IEEE, 2018, pp. 3609--3613.

\bibitem{baik2006eigenvalues}
J.~Baik and J.~W. Silverstein, ``Eigenvalues of large sample covariance
  matrices of spiked population models,'' \emph{Journal of multivariate
  analysis}, vol.~97, no.~6, pp. 1382--1408, 2006.

\bibitem{grenander1958toeplitz}
U.~Grenander and G.~Szeg{\"o}, \emph{Toeplitz forms and their
  applications}.\hskip 1em plus 0.5em minus 0.4em\relax Univ of California
  Press, 1958, vol. 321.

\end{thebibliography}
